\documentclass[sigconf]{acmart}

\AtBeginDocument{%
  \providecommand\BibTeX{{%
    \normalfont B\kern-0.5em{\scshape i\kern-0.25em b}\kern-0.8em\TeX}}}

 \setcopyright{none} 
\renewcommand\footnotetextcopyrightpermission[1]{}

\usepackage{tabularx}
\usepackage{verbatimbox}
\usepackage{multirow}
\usepackage{scalerel}
\usepackage{mathtools}
 \usepackage{url}
\usepackage{circuitikz}
\usepackage{multicol}
\usepackage[noend]{algorithmic}
\usepackage[algoruled,nofillcomment,algo2e]{algorithm2e}
\usepackage{amsmath}
\usepackage{amsfonts}
\usepackage{mdframed}
\usepackage{framed}
\usepackage{cleveref}
\usepackage{hyperref}
\usepackage{balance}
\usepackage{enumitem}
\usepackage{bm}
\usepackage{float}
\usepackage{thmtools, thm-restate}
\usepackage{ifthen}
\usepackage{csquotes}
\usepackage{xcolor}
\usepackage{comment}

\usepackage[utf8]{inputenc} 
\usepackage[T1]{fontenc}    
\usepackage{hyperref}       
\usepackage{url}            
\usepackage{booktabs}       

\usepackage{amsfonts}       
\usepackage{nicefrac}       
\usepackage{microtype}      
\usepackage{lipsum}
\usepackage{enumitem}
\usepackage{textcomp}
\usepackage{subcaption}
\usepackage{tikz}
\usepackage[font=small]{caption}
\usepackage{tikz}

\newcommand{\myparagraph}[1]{\vspace{1ex}\noindent{\bf #1}}
\newcommand{\process}{\texttt{process}}

\usepackage{natbib}
\usetikzlibrary{arrows}
\usepackage{array, stackengine}

\newcolumntype{L}[1]{>{\raggedright\let\newline\\\arraybackslash\hspace{0pt}}m{#1}}
\newcolumntype{C}[1]{>{\centering\let\newline\\\arraybackslash\hspace{0pt}}m{#1}}
\newcolumntype{R}[1]{>{\raggedleft\let\newline\\\arraybackslash\hspace{0pt}}m{#1}}

\newcounter{index}

\usepackage[textsize=footnotesize,color=green!40]{todonotes}

\newcommand{\PP}{\mathbb{P}}

\usepackage{amsthm}

\newcommand{\adria}[1]{\todo[inline,caption={},color=magenta!40]{{\it Adria:~}#1}}
\newcommand{\david}[1]{\todo[inline,color=green!40]{David: #1}}
\newcommand{\james}[1]{\todo[inline,color=blue!40]{James: #1}}
\newcommand{\matt}[1]{\todo[inline,color=red!40]{Matt: #1}}
\newcommand{\nitin}[1]{\todo[inline,color=yellow!40]{Nitin: #1}}

\newcommand{\cheated}{\mathsf{cheated}}
\newcommand{\valid}{\mathsf{valid}}

\newcommand{\inconclusive}{\mathsf{inconclusive}}

\newcommand{\sign}{\mathsf{sign}}

\newcommand{\pvccommit}{\mathsf{pvccommit}}
\newcommand{\assert}{\mathsf{assert}}
\newcommand{\checkfun}{\mathsf{check}}
\newcommand{\out}{\mathsf{output}}

\newcommand{\sgn}{\mathsf{sgn}}
\newcommand{\sk}{\mathsf{sk}}
\newcommand{\pk}{\mathsf{pk}}

\newcommand{\adv}{\mathcal{A}}

\newcommand{\A}{\texttt{P1}}
\newcommand{\B}{\texttt{P2}}
\newcommand{\corrupted}{\texttt{C}}

\newcommand{\prng}{\mathsf{prng}}
\newcommand{\andreduce}{\mathsf{andreduce}}
\newcommand{\xor}{\oplus}
\newcommand{\parity}{\mathsf{parity}}

\newcommand{\startcompact}[1]{\par\vspace{-0.75em}\begin{#1}%
\allowdisplaybreaks\ignorespaces}

\newcommand{\stopcompact}[1]{\end{#1}\ignorespaces}
\makeatletter
\renewenvironment{equation*}{%
    \startcompact{small}
   \mathdisplay@push
   \nonumber
  \mathdisplay{equation*}%
}{%
  \endmathdisplay{equation*}%
  \mathdisplay@pop
  \stopcompact{small}
}
\makeatother

\makeatletter
\renewenvironment{equation}{%
    \startcompact{small}
   \mathdisplay@push
   \nonumber
  \mathdisplay{equation}%
}{%
  \endmathdisplay{equation}%
  \mathdisplay@pop
  \stopcompact{small}
}
\makeatother

\SetKwInOut{Parameter}{Parameters}

\begin{document}
\fancyhead{}
\title{MPC-Friendly Commitments for Publicly Verifiable Covert Security} 

\settopmatter{printfolios=true}

\author{Nitin Agrawal}

\email{nitin.agrawal@cs.ox.ac.uk}
\affiliation{
  \institution{University of Oxford}
  \streetaddress{}
  \city{}
  \country{}
}

\author{James Bell}
\email{jbell@turing.ac.uk}
\affiliation{
  \institution{The Alan Turing Institute}
  \streetaddress{}
  \city{}
  \country{}
}

\author{Adri\`{a} Gasc\'{o}n}
\email{adriag@google.com}
\affiliation{
  \institution{Google}
  \streetaddress{}
  \city{}
  \country{}
}

\author{Matt J. Kusner}
\email{m.kusner@ucl.ac.uk}
\affiliation{
  \institution{University College London}
  \streetaddress{}
  \city{}
  \country{}
}

\begin{abstract}

  We address the problem of efficiently verifying a commitment in a two-party computation. This addresses the scenario where a party P1 commits to a value $x$
  to be used in a subsequent secure computation with another party P2 that wants to receive assurance that P1 did not cheat, i.e. that $x$ was indeed the value inputted into the secure computation.
  Our constructions operate in the publicly verifiable covert (PVC) security model, which is a relaxation of the malicious model of MPC
  appropriate in settings where P1 faces a reputational harm if caught cheating.
  
  We introduce the notion of PVC commitment scheme and indexed hash functions to build commitments schemes tailored to the PVC framework, and propose constructions for both arithmetic and Boolean circuits
  that result in very efficient circuits.
  From a practical standpoint, our constructions for Boolean circuits are $60\times$ faster to evaluate securely, and use $36\times$ less communication than baseline
  methods based on hashing. Moreover, we show that our constructions are tight in terms of required non-linear operations, by proving lower bounds on the nonlinear gate count of commitment verification circuits.
  Finally, we present a technique to amplify the security properties of our constructions that allows to efficiently recover malicious guarantees with statistical security.

\end{abstract}


\newcommand{\submission}{T} 
\newcommand{\nsubmission}{F} 

\ifthenelse{\equal{\submission}{T}}
  {\renewcommand{\matt}[1]{}
 \renewcommand{\james}[1]{}
 \renewcommand{\david}[1]{}
 \renewcommand{\nitin}[1]{}
 \renewcommand{\adria}[1]{}}

\twocolumn

\maketitle

\section{Introduction}\label{sec:intro}

Secure multi-party computation (MPC) methods have undergone impressive improvements in the last decade. Advances in the scalability of garbled circuit protocols \cite{wang2017authenticated, wang2017global, ben2017efficient}, commitment schemes \cite{frederiksen2018committed}, and oblivious transfer \cite{naor2001efficient, asharov2013more} have transformed the range of applications for MPC~ \cite{lindell2020secure}. In particular, a significant amount of research efforts have been recently devoted to finding efficient MPC protocols for training and evaluation of widely-deployed machine learning (ML) models \cite{nikolaenko_privacy-preserving_2013, petsregression, agrawal2019quotient, mohassel2017secureml, mohassel2018aby, wagh2019securenn, cryptonets, tapas}. These works enable collaborative training, as well as private predictions, where users can get predictions from a confidential model while preserving the privacy of their input.

Generally speaking, the guarantee of an MPC computation is that the inputs of the participants remain private to other parties, but that does not prevent parties to choose their inputs in an arbitrary way, e.g., in the well-known millionaires problem, nothing prevents a millionaire from lying about their wealth. Going back to the private prediction application mentioned above: nothing prevents the model owner from modifying the model arbitrarily. This is a problem in settings where the model has to satisfy certain non-functional constraints such as safety, fairness, or privacy. These constraints undermine accuracy (as often measured in ML) and thus the model owner may have an incentive to switch the model. This exact problem and, more generally, model certification, was tackled recently by~\citet{kilbertus2018blind} and \citet{segal2020fairness}. Both these works show that commitments verified in MPC can help here. For example, consider a service provider offering dietary or exercise recommendations based on personal data. Users may require the service provider to commit to a recommendation algorithm that is certified not to make harmful recommendations (which could have been verified and signed by a regulator). More formally a user requires the following 2-party functionality: the service provider (P1) commits to an input $x$ by producing commitment $c$ and sends $c$ to the user (P2). Later, P2 uses $c$ to verify that $x$ is being used by P1 inside an MPC protocol between both parties. We call this \emph{MPC on committed data}. We describe an application of this framework to certified predictions in Appendix~\ref{app:fairness}, along with an empirical example on realistic data showing that heuristic methods that do not ensure that the model does not change will fail. Concretely, we show that changing a single parameter in a fair model results in an unfair model with increased accuracy. 

So far current work on this uses standard collision-resistant hash functions such as SHA-256 \cite{kilbertus2018blind} and SHA-3 \cite{segal2020fairness} to produce and then verify commitments in MPC. However, these methods do not take advantage of two key properties of this setting: 1. \emph{Interactivity}: given that an MPC protocol needs to be run between the user and service provider to compute some functionality (e.g., a recommendation), it is possible to leverage the interactivity of the protocol to construct a commitment; 2. \emph{Reputation of service provider}: as this computation involves a service provider who relies on users for profit, a protocol can be constructed so that cheating would harm the reputation of the service provider, using ideas from Publicly-Verifiable Covert (PVC) security  \cite{DBLP:conf/asiacrypt/AsharovO12}. 

Based on these properties we make a simple observation: one can detect if an input $x$ to a Boolean MPC protocol has been changed with probability $\nicefrac{1}{2}-\epsilon$ (for arbitrarily small $\epsilon>0$) using a simple additional Boolean circuit as part of the protocol (more details on such circuits are in Figure~\ref{fig:constructions-1-2-3}). The idea is that in MPC a hash can be efficiently constructed by using inputs from both parties, an idea we call \emph{indexed hash functions}. These functions allow one to build MPC commitments in the PVC setting that analytically and experimentally outperform prior approaches by as much as $60\times$ in runtime and $36\times$ in communication.

\myparagraph{Other Related Work. }
\citet{baum2016garbling} also discusses the problem of input validity, using efficient SFE protocols. In particular, the solution utilizes universal hash functions and committed OT. The protocol specifically improves the performance of garbled circuit based secure function evaluation for cases where sub-circuits depend on only one party's input. Concurrent with \cite{baum2016garbling}, \citet{katz2016efficiently} propose a solution performing predicate checks followed by secure evaluation of an arbitrary function, if the inputs pass the verification. However, both these works focus on malicious security, different from our proposal in PVC security model, utilizing properties 1,2. In particular, our approach does not fit their paradigm because our commitment verification has a private input. However, our results for maliciously secure setting (section \ref{sec:malicious}) are compatible with their optimizations.

In this paper we give a technical overview of our paper, and describe our contributions. We then introduce indexed hash functions and give efficient constructions for them. We describe how to use these hash functions to enable MPC on committed data. We give an analytic and experimental comparison with prior work.
We derive lower bounds on the number of AND gates necessary for indexed hash functions, demonstrating that some of our constructions are as efficient as possible. A natural question is if the security guarantees of indexed hash functions can be extended to computational security. We answer this question affirmatively: we give a construction and present complexity results. Finally we describe initial ideas of extensions of this approach for arithmetic circuits.

Although we chose to motivate our contribution from the perspective of certified predictions, our results are general, and essentially provide constructions of commitment schemes tailored for PVC security, along with efficient implementations in MPC.
\section{Preliminaries}\label{sec:prelims}
We give a brief background on key ideas we will use in the paper.

\myparagraph{Hash functions and pseudo-randomness.}
We consider a hash function to be a function $h:\{0,1\}^*\rightarrow O$ for some finite output space $O$. Informally, $h$ is \emph{collision resistant} if no adversary is capable of finding two distinct inputs with the same image except with negligible probability. The formal definition requires talking about families of hash functions~\cite{collisionresistance}. A pseudo-random number generator, or PRNG, is a function $\prng:K\times \mathbb{N}\rightarrow \{0,1\}^*$ which maps $(k,b)$ to a bit-string of length $b$. Both hash functions and PRNGs can be modelled as random oracles i.e. as a uniformly random mapping from their inputs to their outputs.

\myparagraph{Publicly-verifiable covert (PVC) security.} Covert security \cite{DBLP:journals/joc/AumannL10} weakens the malicious security setting by guaranteeing that a cheating party (who may behave arbitrarily) will be caught by the other party with a probability, $p$, referred to as the covert security parameter. The motivation for covert security is that if certain parties have a reputation to preserve, then the risk associated with being caught, outweighs the benefit of cheating. This allows faster protocols than malicious security \cite{DBLP:journals/joc/AumannL10, DBLP:conf/eurocrypt/GoyalMS08, DBLP:conf/tcc/DamgardGN10, DBLP:conf/crypto/Lindell13}. PVC security was introduced by \citet{DBLP:conf/asiacrypt/AsharovO12}. It, with probability $p$, provides a publicly-verifiable proof of cheating, which allows greater reputational harm and possibly legal repercussions for a cheater.

\section{Technical Overview}\label{sec:overview}

\begin{figure}[t!]
  \includegraphics[width=0.9\columnwidth]{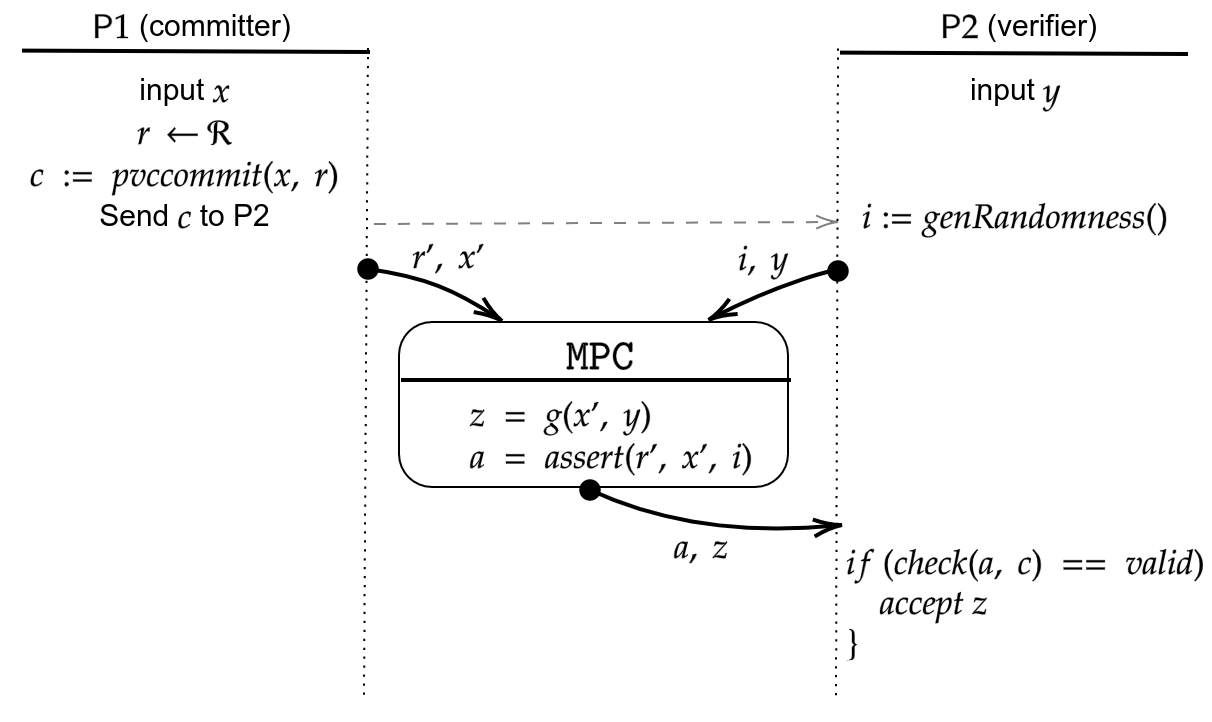}
  \caption{The diagram shows a multi-party computation with a committed input $x$, as enabled by our constructions. Party $1$ (the committer) holds an input $x$, for which it generates a commitment $c$ and sends it to party $2$ (the verifier). The commitment is randomized using $r$ to ensure privacy for $x$, i.e., that the commitment is hiding. Later on, the parties engage in a secure computation of a generic function $g$, where party $1$ inputs $x'$, and party $2$ inputs input $y$. For the purpose of verifying that $x = x'$, party $2$ derives a challenge $i$ from $c$, and the MPC returns a certificate $a$ that can be checked by P2. This guarantees to party $2$ that $g$ is evaluated on the value $x$ to which party $1$ had previously committed.}
\label{fig:basic_diagram}
\end{figure}

Hashing is a useful primitive to implement in MPC, as it enables privacy-preserving consistency checks, and thus {\em MPC on committed data}.
This paper focuses on efficient implementation of this functionality, depicted abstractly in Figure~\ref{fig:basic_diagram}. The high-level goal is to enable a party P1 to commit to a {\em private} value $x$ at some point in time and later, when engaging in a secure computation of a function $g$ with a second party P2, provide the assurance to P2 that P1 inputs the committed value $x$ into the computation and not some other value. This is modelled, analogously to commitment schemes, by three algorithms: $\pvccommit$, $\assert$, and $\checkfun$. As shown in Figure~\ref{fig:basic_diagram}, $\pvccommit$ outputs a commitment to an input, later in the MPC, $\assert$ is run. As we will see next, naively, $\assert$ could simply compute $\pvccommit$, but our constructions will leverage randomness in  $\assert$. Otherwise the algorithms correspond to a standard commitment scheme, in that they should satisfy the standard {\em binding and hiding} properties. The function $\checkfun$ is used to interpret the results and see whether cheating has occurred. The verification is split between $\assert$ and $\checkfun$ and a single commitment allows for arbitrarily many secure computations, in the future.

The functionality of Figure~\ref{fig:basic_diagram} has appeared in previous works on {\em certified predictions}~\cite{segal2020fairness,kilbertus2018blind} where $x$ is a confidential machine learning model owned by P1 that has been checked by a certifying authority to have certain properties, e.g., fairness. In those applications, $g$ corresponds to model evaluation and users receive predictions using the certified model $x$. MPC on committed data directly enables this functionality, with the model owner and the users playing the roles of parties P1 and P2 in Figure~\ref{fig:basic_diagram}, respectively. We describe and motivate this application in Appendix~\ref{app:fairness}.

\myparagraph{Baseline protocol.}To see how a hash function $h$, e.g., SHA3 in practice, can be used to implement MPC on committed data, consider the following instantiation of $\pvccommit, \assert$, and $\checkfun$. P1 can just choose a random $r$ and have $\pvccommit(x,r)$ and $\assert(x,r)$ be $h(x || r)$. Then, $\checkfun$ just checks that they are equal. Both $\pvccommit$ and $\checkfun$ are efficient in this instantiation, so one would want $h$ to have an efficient MPC protocol.

\myparagraph{MPC-friendly hashing.}
The works on fairness certification of \citet{segal2020fairness} and \citet{kilbertus2018blind} propose the above baseline construction.
Segal et. al concretely instantiate $h$ with the Keccak-F function, which is the basis of the SHA3 standard.
That function takes a $1600$ bit input and can be implemented
by a Boolean circuit of $38400$ AND gates i.e. $24$ AND gates per input bit~\cite{segal2020fairness, esat}. Using a Merkle tree for succinctness of $c$, the total number of hashes is twice the number of input blocks, resulting in $48$ AND gates per input bit of $x$. Thus $\assert$ would result in an overhead of $48$ AND gates/bit in this instantiation. Alternatively, using SHA3-256 in sponge mode results in an overhead of roughly $35$ AND gates/bit. On the other hand, using an MPC-optimized hash (but new and susceptible~\cite{dinur2015optimized}) LowMCHash-256~\cite{albrecht2015ciphers} in sponge mode roughly take up 14 AND gates/bit. Note that AND gate counts, and non-linear gate counts in general, are a standard reference for computation time in MPC, and secure computation in general. In this work we propose efficient MPC-friendly commitments schemes based on hashing, with a focus on Publicly Verifiable Covert (PVC) security.

Our starting observation is that executing a collision-resistant hash function such as SHA-256 in a PVC-secure protocol is an overkill: Note that ensuring that commitments are binding, i.e., that P1 in Figure~\ref{fig:basic_diagram} can not generate $x', r'$ such that $x\neq x'$ and the verification passes, up to negligible probability, but then ensuring that the subsequent MPC is secure only up to a constant probability $p$ leaves some potential room for weakening the binding guarantee of the commitment scheme to favor efficiency.
To take advantage of the PVC setting, we observe that $\assert$ must in some way receive randomness from P2, as it is against this randomness that P1 will have probability $p$ of being caught. We design $\assert$ function that leverage this observation, and result in an overhead of the $\assert$ circuit as low as {\em half an AND gate per bit of the input $x$}, even when $p$ is close to $1$, including $1-2^{-\sigma}$ for a statistical security parameter $\sigma$.

\subsection{Contributions}

We introduce the notion of an {\em indexed hash function}. This, roughly speaking, is a function that produces a hash of an input $x$, given a random value $r$ and an index $i$ from a domain $\mathcal{I}$. The index of the hash function plays the role of the randomness chosen by P2 mentioned above. We build indexed hash functions from any collision-resistant hash function $h$ and prove properties related to collision resistance that allow us to construct PVC commitments from indexed hash functions, and use them for achieving the functionality in Figure~\ref{fig:basic_diagram} with PVC security.

Given an indexed hash function $H$, our proposed PVC commitment schemes follow the following high-level structure:
$\pvccommit$ computes $\left( H(j, r, x)\right)_{j\in I}$, i.e. the hash evaluated at all indices, the verifier selects an index $i\in \mathcal{I}$, $\assert$ computes $H(i,r,x)$, and $\checkfun$ checks that this value is correct. This check fails (in the sense of giving a false positive) with a probability upper bounded by  $1-p$ ($p$ is called the covert security parameter). This fact is formalized by reducing an appropriate notion of collision resistance of $H$ to the collision resistance of $h$.

The efficiency of our scheme relies on the fact that our constructions for indexed hash functions
are very efficient to be evaluated in MPC, requiring a very small number of XOR and AND gates
with respect to the input size and thus inducing a very small overhead
when evaluated in MPC. Next, we summarize the organization of the paper,
highlighting the key contributions of each section.

\myparagraph{A construction for Boolean circuits (Section \ref{sec:coverthashes}).}
We give a construction for Boolean circuits that can achieve a covert security parameter arbitrarily close to $1/2$, which asymptotically requires only half an AND gate per bit. Furthermore, in practice, it requires less than one AND gate per bit for moderately sized inputs.

\myparagraph{PVC (Section \ref{sec:integritycheck}).}
We define PVC commitment schemes and their security properties, propose
a secure instantiation based on indexed hash functions, and show how to use them for committed PVC MPC.

\myparagraph{Experimental Evaluation. (Section \ref{sec:experiments})}
We fully implement our most practical construction, and compare it with the
baseline approach (both instantiated with SHA3 and an MPC-friendly hash function LowMCHash). Our experiments show a $60\times$ speed up and $36\times$ less communication in the resulting $\assert$ compared to SHA3. It takes $\sim1$ MB and $<15$ minutes to commit million bit inputs. 

\myparagraph{Lower bounds (Section \ref{sec:lower}).}
We show the optimality of our construction by proving lower bounds for both our approach and the baseline hash based approach. In particular, we show that if the resulting hash is required to be even remotely succinct (to have size $<99x/100$) then (i) the indexed hash function must require at least half an AND gate per bit asymptotically, and (ii) with an ordinary hash function at least one AND gate per bit is required. We also show that even if we remove the assumption on the size of the hash at least $1/5$ of an AND gate per bit would be required.

\myparagraph{A way to amplify security (Section \ref{sec:malicious})}.
We give a way to achieve covert security parameter $1-2^{-\sigma}$, and thus full statistical security, in both of the Boolean and arithmetic cases. While the method relies on repeating our constructions to amplify their probabilistic guarantees, we are able to maintain the asymptotic half a non-linear gate per bit. This result is mostly of theoretical relevance at the moment, as it does not beat the state of the art of MPC-friendly hash functions for input sizes that are currently practical.

\myparagraph{A construction for Arithmetic circuits (Section \ref{sec:arithmetic}).}
We give an analog to our Boolean construction for arithmetic circuits which can achieve a covert security parameter arbitrarily close to $1$ with only half a MULT gate per input element asymptotically. This is also a practical improvement over the best known ordinary hash functions~\cite{albrecht2015ciphers}. However the security parameter can not be taken to be $1$ minus negligible.

\section{Indexed Hash Functions}\label{sec:coverthashes}

\label{sec:defs}

In this section we will introduce the primitive we will use to build our commitments: \emph{indexed hash functions}. We do so informally and then formally, and give some examples of indexed hash functions. The examples will show that secure indexed hash functions have much smaller circuits than ordinary secure hash functions.

Like an ordinary secure hash function, an indexed hash function takes an input from some space $\mathcal{X}$ and produces an output in a space $\mathcal{O}$. 
When working with a specific $x\in \mathcal{X}$ we will denote by $n$ the bitlength of $x$.

The whole idea is to take advantage of the fact that the verifier (P2) can have an input to the indexed hash.
We call this input the {\em index} of the hash function and denote it by $i$, drawn from an index set $\mathcal{I}$. If the wrong input from $\mathcal{X}$ is used during verification, then at least a fixed fraction of indices $i$ will result in an incorrect hash.

We need to ensure that the hash is hiding, otherwise if $x_1,x_2\in \mathcal{X}$ were the possible inputs by the committer, then by computing the hash of $x_1$ and $x_2$ the verifier could learn which was the true input from the output. Thus the committer must provide an extra random input $r$ from some set $\mathcal{R}$ containing enough entropy to hide the true input. Hence we define indexed hash functions
as functions taking an index $i$, a random nonce $r$, and an input $x$.\footnote{We call these hash functions because they have output smaller than their input and, even with the insertion of randomness that is not technically part of the function, a single evaluation of this function would not create a commitment.}

\begin{definition}
An \emph{indexed hash function} is a function $H:\mathcal{I}\times \mathcal{R}\times \mathcal{X} \rightarrow \mathcal{O}$.
\end{definition}

\subsection{Collision resistance}

Next, we formally define two properties for indexed hash function that we require:
(i) the hiding property, i.e., not leaking information about $x$, and (ii) a notion of collision resistance which we call \emph{$q$-collision boundedness}.
For (i) we can use the definition from classical commitment schemes: we say $H$ is \emph{hiding}
if an adversary learns a negligible amount about an input $x$ from learning the value of $H(i,r,x)$
for all $i \in \mathcal{I}$, and a uniformly random (and secret from the adversary) $r\in \mathcal{R}$.

For (ii), we start by defining the concept of a $q$-collision, which denotes
a pair of inputs on which $H$ collides for at least a fraction $q$ of all possible indices $|\mathcal{I}|$.

\begin{definition}
Let $q\in [0,1]$. We call a quadruple $r,x,r',x'$ a \emph{$q$-collision} of $H$ if $x\neq x'$ and
\begin{equation*}
|\{i\in \mathcal{I} | H(i,r,x) = H(i,r',x')\}| \geq q|\mathcal{I}|
\end{equation*}
\end{definition}

We can now define our notion of collision resistance.
Informally, $H$ is \emph{$q$-collision resistant},
if adversaries are unable to find a $q$-collision of $H$ except with negligible probability.
We formalize this using families of indexed hash functions, in turn indexed by a key $k\in K$ generated by a generator $G$ taking a computational security parameter $\lambda$. 
This is similar to the standard definition of a family of collision-resistant hash functions.
Moreover, we say that $H$ is \emph{$q$-collision bounded} if it is $q'$-collision resistant for every $q'>q$.

\myparagraph{The security parameter $\lambda$.} 
We use a single computational security parameter $\lambda$ for all aspects of our constructions. This includes their underlying collision resistant hash function,
as well as the size of the source of randomness $\mathcal{R}$ and the set of indices $\mathcal{I}$.
In particular, $|\mathcal{I}|$ is polynomial in $\lambda$ and  $|\mathcal{R}|$ is exponential in $\lambda$ in all constructions.
Thus our security is formalized in terms of polynomial time adversaries w.r.t. $\lambda$, and whose advantage should be bounded by a negligible function in $\lambda$.
Note that this implies that an attacker is allowed to iterate over $\mathcal{I}$, and in fact in practice we will ensure that $\mathcal{I}$ is as small as possible for efficiency.

\begin{definition}
Given a generator $G$, security parameter $\lambda$, and key $k=G(\lambda)$, a family $\{H_k\}_{k\in K}$ is $q$-collision resistant if, for any probabilistic polynomial time algorithm $A$ we have that 
\begin{equation*}
\PP[A(k)\text{ is a $q$-collision of }H_k]< \text{negl}(\lambda)
\end{equation*}
\end{definition}

Given this, we can define our main notion of collision as follows.

\begin{definition}
A family $\{H_k\}_{k\in K}$ is $q$-collision bounded if it is $q'$-collision resistant for all $q'>q$.
\end{definition}

Any family $\{H_k\}_{k\in K}$ that is $q$-collision resistant is also $q$-\emph{collision bounded}. This is simply due to the definition of $q$-collisions: any $q$-collision is also a $q'$-collision for all $1 \geq q' > q$.

This property will be useful later because by choosing an index uniformly at random we can distinguish between any two inputs $x, x' \in \mathcal{X}, \mbox{s.t.}$ with probability at least $1-q$ by looking at a hash. We can now make the hiding property precise. That $H$ is hiding will be proved under the assumption that $h$ is a random oracle.

\begin{definition}\label{def:hiding}
A family of indexed hash functions $\{H_k\}_{k\in K}$ is hiding if for any polynomial time algorithm $A$ and any $x,x'\in \mathcal{X}$, for a uniformly random choice of $r \in \mathcal{R}$ we have
\begin{equation*}
\begin{split}
\PP(A&(k,(H_k(i,r,x))_{i\in\mathcal{I}})=1)=\\
&\PP(A(k,(H_k(i,r,x'))_{i\in\mathcal{I}})=1)+\text{negl}(\lambda)
\end{split}
\end{equation*}
\end{definition}

As mentioned above, we will construct our indexed hash functions by building them from an ordinary secure hash function $h$. We do so because it will allow us to prove $q$-collision resistance of $H$ via a (ptime) reduction to collision resistance of $h$. To argue about collision boundedness we formalize the notion of a construction.
\begin{definition}
A \emph{construction} of an indexed hash function is a function $\mathcal{C}$ which given a hash $h$ and the security parameter $\lambda$, returns an indexed hash $H$.
\end{definition}

We say that $\mathcal{C}$,
\emph{preserves $q$-collision boundedness} if there is an (efficient) algorithm which, given a $q$-collision of $\mathcal{C}$, returns a collision of $h$.
Thus if a powerful adversary is unable to find a collision in some fixed hash function $h$, then it is reasonable to assume they cannot find a $q'|\mathcal{I}|$-sized collision in $\mathcal{H}$ for any $q' > q$.

\begin{definition}
We say that $\mathcal{C}$ \emph{preserves $q$-collision boundedness} if $\{h_k\}_K$ being collision resistant implies that $\{H_k=\mathcal{C}(h_k)\}_K$ is $q$-collision bounded.

\end{definition}

\subsection{Constructions}

We will now give four constructions of indexed hash functions denoted $\mathcal{C}_0$ through $\mathcal{C}_3$. Our construction 3 is the most practical and the one we would recommend to use (we will use this construction in our experiments in Section~\ref{sec:experiments}), but we include all four to build up ideas incrementally.
Our goal is to derive indexed hash functions that 1. {\em Are efficient to implement in MPC}: this aspect of the constructions in this section is captured in terms of {\em the number of AND and XOR gates} of their corresponding Boolean circuit implementations (in Section~\ref{sec:arithmetic} we give a construction for arithmetic circuits); 2. {\em Have a small index domain $\mathcal{I}$}: this directly corresponds to the commitment size of the PVC commitment schemes that we will build on top of them. For these constructions the value of $q$ (i.e., for which constructions are $q$-collision bounded) is one minus the security parameter of the PVC commitment scheme derived later. We summarize computation, size, and allowed $q$ values in Table~\ref{tab:asymptotics-small}. In this section we have $q = 1/2$ and  $q = 1/2+\epsilon$ (for arbitrarily small $\epsilon > 0$),
but in Section~\ref{sec:malicious} we show how to achieve arbitrarily small $q$.

\begin{table}
\small
\resizebox{0.95\columnwidth}{!}{\
\begin{tabular}{ ccccc } 
\toprule
 \multirow{2}{*}{\bf Construction} &  \multicolumn{1}{p{10ex}}{\bf\centering $\bm{\mathcal{C}_0}$ \\ (baseline)} &  \multirow{2}{*}{$\bm{\mathcal{C}_1}$} &  \multirow{2}{*}{$\bm{\mathcal{C}_2}$}  &  \multicolumn{1}{p{10ex}}{\bf\centering $\bm{\mathcal{C}_3}$ \\ (main)}  \\ \midrule
 \bf \# ANDs & $C_Nn$ & $n+C_N\frac{n}{b}$ & $\frac{n}{2}+C_N\frac{n}{b}$ & $\frac{n}{2}+C_N\frac{n}{b}$  \\[1ex]
 \bf \# XORs & $C_Ln$ & $n+C_L\frac{n}{b}$ & $\frac{3n}{2}+C_L\frac{n}{b}$ & $\frac{3n}{2}+C_L\frac{n}{b}$ \\[1ex]
 $\bm{|\mathcal{I}|}$ & $1$ & $2^b$ & $2^b$ & $\frac{b+\lambda+1}{2\epsilon^2}$ \\[1ex]
 $\mathbf{q}$ & $1$ & $\frac{1}{2}$ & $\frac{1}{2}$ & $\frac{1}{2}+\epsilon$ \\
\bottomrule
\end{tabular}
}
\caption{\label{tab:asymptotics-small} Let $H$ be an indexed hash function resulting from a construction, using $h$ as the underlying collision resistant hash function. This table shows (i) the size of a Boolean circuit for $H$ in terms of number of bits $n$ of the input (omitting lower order $o(n)$ additive terms), where $C_N$ and $C_L$ denote the number of $AND$ and $XOR$ gates of $h$, respectively;  (ii) the size of $\mathcal{I}$, and (iii) the value of $q$ for which $q$-collision resistance holds. The parameter $b$ can be chosen to be any even positive integer and $\epsilon$ can be taken to be any positive value. }

\end{table}

\myparagraph{Blueprint for our constructions.}
Consider an indexed hash function $H$ taking an index $i$, randomness $r$, and input $x$.
All our constructions are parametrized by a {\em block size} $b\in [|x|]$ and a {\em block digest function} $d$.
The latter takes (i) a binary encoding of $i$ and (ii) a bitstring of size $b$,
and outputs a {\em single} bit, i.e. $d : \mathcal{I} \times \{0, 1\}^b\mapsto \{0,1\}$.
The indexed hash function is defined to be the result of
\begin{enumerate}
\item splitting $x$ into $n/b$ consecutive blocks of $b$ bits
(if $|x|$ is not a multiple of $b$, it can be padded with zeros),
\item applying $d(i, \cdot)$ to each block $x_j$,
and
\item outputting the length $n/b$ bitstring resulting from concatenating all digested bits $d(i, x_j)$.
\end{enumerate}

We denote the result of processing an input $i, x$ as in steps $1$-$3$
by $\process_{b, d}(i, x)$. Digest function $d$ determines the size of the index set $\mathcal{I}$.

Then, each construction $\mathcal{C}$ is defined by a block size $b$, digest function $d$, an ordinary collision resistant hash function $h$, and a set of random masks $\mathcal{R}$ (which we always take to be $\{0,1\}^\lambda$) as:
\begin{align}\label{eq:basicform}
\mathcal{C}(h,\lambda)(i,r,x) = h(r||i||\process_{b, d}(i, x))
\end{align}
When presenting the three constructions in this section we will denote the digest function associated with $\mathcal{C}_j$ by $d_j$.
The motivation behind this presentation is simplicity, as it is now enough to define $d_1, d_2, d_3$, and
the arguments in our proofs only need to refer to $d_j$.

\myparagraph{The hiding property.} We can prove the hiding property (Def.~\ref{def:hiding}) without knowing anything about $d$ and thus the size of $\mathcal{I}$, so we do this in generality for all the constructions.

\begin{restatable}[]{theorem}{hidingtheorem}
Suppose $\mathcal{C}$ is given by Equation~\ref{eq:basicform}. If $\{h_k\}_{k\in K}$ is a family of random oracles then $\{\mathcal{C}(h_k,\lambda)\}_{k\in K}$ is hiding.
\end{restatable}
See Appendix~\ref{app:coverthashes} for the proof of this theorem.

\myparagraph{Collision boundedness.} For each $C_j$ we propose, we will show $q$-collision boundedness of $H=\mathcal{C}_j(h,\lambda)$. This will be done by showing that for at most $q|\mathcal{I}|$ indices $i$ we have $\process_{b, d}(i,x)=\process_{b, d}(i,x')$. That this is a sufficient condition for $q$-collision boundedness is shown in the following theorem.

\begin{theorem}\label{thm:collision}
Let $q\in [0, 1]$. Suppose that for any $x\neq x'$ and for any sufficiently large $\lambda$,
{\small
\begin{equation*}
|\{i \in \mathcal{I} | \process_{b, d}(i,x)=\process_{b, d}(i,x')\}|\leq q|\mathcal{I}|
\end{equation*}}\noindent
then the construction in Equation~\ref{eq:basicform} preserves $q$-collision boundedness.
\end{theorem}

\begin{proof}
Let $\{h_k\}_{k \in K}$ be a collision resistant family of hash functions, and $H_k=\mathcal{C}(h_k,\lambda)$.  Suppose that the hypothesis of the statement holds but there exists a polynomial time algorithm $A$ which finds a $q$-collision in $H_k$ with non-negligible probability. We show next that $\{h_k\}_{k \in K}$ is not collision resistant: a contradiction.

Let $k=G(\lambda)$. For sufficiently large $\lambda$, a probabilistic ptime (in $\lambda$) algorithm B for finding a collision in $h$ with non-negligible probability is given by the following. Given $k$, $B$ computes $(r, x,r',x')=A(k)$. If $(r, x,r',x')$ is a $q'$-collision for some $q'>q$ (this happens with non-negligible probability) then $B$ computes an index $i$ such that $\process_{b, d}(i,x)\neq \process_{b, d}(i,x')$ but $H_k(i,r,x) = H_k(i,r',x')$. Note that such an $i$ must exist with probability $1$ and can be found by exhaustion in time $O(\mathcal{I})$ (and thus $O(\lambda)$.
By Equation~\ref{eq:basicform} we now have that $r||i||\process_{b, d}(i,x)$ and $r'||i||\process_{b, d}(i,x')$ form a collision in $h_k$. The algorithm $B$ outputs this collision. The fact that $B$ succeeds with non-negligible probability contradicts the collision resistance of the family $\{h_k\}_{k \in K}$.
\end{proof}

We now proceed to present each construction $\mathcal{C}_j$ by specifying the length of $\mathcal{I}$ and the digest function $d_j$ to sub into Equation~\ref{eq:basicform}.
Recall that all constructions are summarized in Table~\ref{tab:asymptotics-small}.

\myparagraph{Construction $0$.}
Firstly let us consider a trivial construction. Let $\mathcal{I}$ be the set containing only the empty string $\epsilon$ and let $d_0(i,x) = x$ and the block size be $b_0 = |x|$.
Let
\begin{align*}
h(r || \epsilon || \process_{b_0, d_0}(i,x)) = h(r || x).
\end{align*}

\begin{theorem}
$\mathcal{C}_0$ preserves $0$-collision boundedness.
\end{theorem}
\begin{proof}
By Theorem~\ref{thm:collision} this is immediate as the identity function has no collisions, and thus the condition of that theorem holds for $q = 0$.
\end{proof}

This construction is very simple and $d_0$ is trivial to compute. However,
we can improve over this by making the digest function $d$ compress the input so that $h$ only has to be computed on an input much smaller than $x$, resulting in a more efficient circuit.

\myparagraph{Construction $1$: $n$ AND gates.}
The digest function for construction $1$ is shown in Figure~\ref{fig:constructions-1-2-3} (left). As above, this construction is parameterised by a block size $b$. Let $\mathcal{I}=\{0,1\}^b$ and let $x_j$ be the block containing bits $jb$ through $jb+b-1$, inclusive, of $x$ (padding $x$ with zeros to length a multiple of $b$). We use $\&$ to denote bitwise AND, and let $\parity$ map bit strings to the XOR of all their bits. Then we define $d_1(i, x_j) = \parity(x_j \& i)$.

\tikzset{every picture/.style={line width=0.75pt}} 

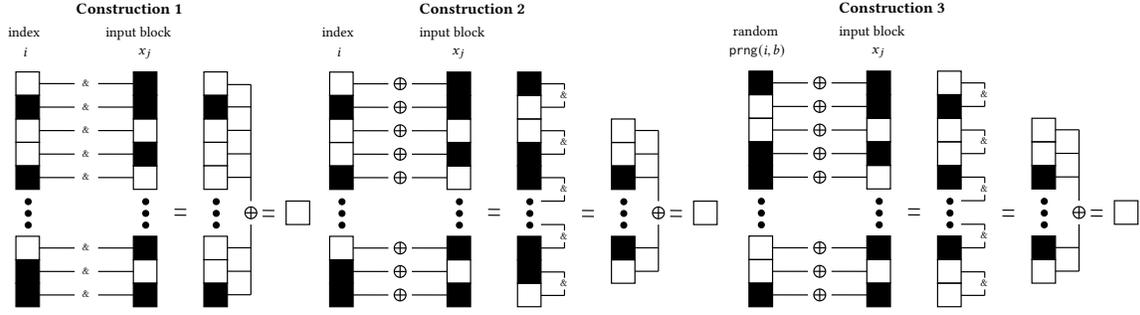
\begin{figure*}[th!]
\resizebox{0.23\textwidth}{!}{%
\begin{tikzpicture}[x=0.75pt,y=0.75pt,yscale=-1,xscale=1]

\draw   (100,40) -- (120,40) -- (120,60) -- (100,60) -- cycle ;
\draw  [fill={rgb, 255:red, 0; green, 0; blue, 0 }  ,fill opacity=1 ] (100,60) -- (120,60) -- (120,80) -- (100,80) -- cycle ;
\draw   (100,80) -- (120,80) -- (120,100) -- (100,100) -- cycle ;
\draw   (100,100) -- (120,100) -- (120,120) -- (100,120) -- cycle ;
\draw  [fill={rgb, 255:red, 0; green, 0; blue, 0 }  ,fill opacity=1 ] (100,120) -- (120,120) -- (120,140) -- (100,140) -- cycle ;
\draw   (100,180) -- (120,180) -- (120,200) -- (100,200) -- cycle ;
\draw  [fill={rgb, 255:red, 0; green, 0; blue, 0 }  ,fill opacity=1 ] (108,150.5) .. controls (108,149.12) and (109.12,148) .. (110.5,148) .. controls (111.88,148) and (113,149.12) .. (113,150.5) .. controls (113,151.88) and (111.88,153) .. (110.5,153) .. controls (109.12,153) and (108,151.88) .. (108,150.5) -- cycle ;
\draw  [fill={rgb, 255:red, 0; green, 0; blue, 0 }  ,fill opacity=1 ] (108,160.5) .. controls (108,159.12) and (109.12,158) .. (110.5,158) .. controls (111.88,158) and (113,159.12) .. (113,160.5) .. controls (113,161.88) and (111.88,163) .. (110.5,163) .. controls (109.12,163) and (108,161.88) .. (108,160.5) -- cycle ;
\draw  [fill={rgb, 255:red, 0; green, 0; blue, 0 }  ,fill opacity=1 ] (108,170.5) .. controls (108,169.12) and (109.12,168) .. (110.5,168) .. controls (111.88,168) and (113,169.12) .. (113,170.5) .. controls (113,171.88) and (111.88,173) .. (110.5,173) .. controls (109.12,173) and (108,171.88) .. (108,170.5) -- cycle ;
\draw  [fill={rgb, 255:red, 0; green, 0; blue, 0 }  ,fill opacity=1 ] (100,200) -- (120,200) -- (120,220) -- (100,220) -- cycle ;
\draw  [fill={rgb, 255:red, 0; green, 0; blue, 0 }  ,fill opacity=1 ] (100,220) -- (120,220) -- (120,240) -- (100,240) -- cycle ;
\draw  [fill={rgb, 255:red, 0; green, 0; blue, 0 }  ,fill opacity=1 ] (200,40) -- (220,40) -- (220,60) -- (200,60) -- cycle ;
\draw  [fill={rgb, 255:red, 0; green, 0; blue, 0 }  ,fill opacity=1 ] (200,60) -- (220,60) -- (220,80) -- (200,80) -- cycle ;
\draw   (200,80) -- (220,80) -- (220,100) -- (200,100) -- cycle ;
\draw  [fill={rgb, 255:red, 0; green, 0; blue, 0 }  ,fill opacity=1 ] (200,100) -- (220,100) -- (220,120) -- (200,120) -- cycle ;
\draw   (200,120) -- (220,120) -- (220,140) -- (200,140) -- cycle ;
\draw  [fill={rgb, 255:red, 0; green, 0; blue, 0 }  ,fill opacity=1 ] (200,180) -- (220,180) -- (220,200) -- (200,200) -- cycle ;
\draw  [fill={rgb, 255:red, 0; green, 0; blue, 0 }  ,fill opacity=1 ] (208,150.5) .. controls (208,149.12) and (209.12,148) .. (210.5,148) .. controls (211.88,148) and (213,149.12) .. (213,150.5) .. controls (213,151.88) and (211.88,153) .. (210.5,153) .. controls (209.12,153) and (208,151.88) .. (208,150.5) -- cycle ;
\draw  [fill={rgb, 255:red, 0; green, 0; blue, 0 }  ,fill opacity=1 ] (208,160.5) .. controls (208,159.12) and (209.12,158) .. (210.5,158) .. controls (211.88,158) and (213,159.12) .. (213,160.5) .. controls (213,161.88) and (211.88,163) .. (210.5,163) .. controls (209.12,163) and (208,161.88) .. (208,160.5) -- cycle ;
\draw  [fill={rgb, 255:red, 0; green, 0; blue, 0 }  ,fill opacity=1 ] (208,170.5) .. controls (208,169.12) and (209.12,168) .. (210.5,168) .. controls (211.88,168) and (213,169.12) .. (213,170.5) .. controls (213,171.88) and (211.88,173) .. (210.5,173) .. controls (209.12,173) and (208,171.88) .. (208,170.5) -- cycle ;
\draw   (200,200) -- (220,200) -- (220,220) -- (200,220) -- cycle ;
\draw  [fill={rgb, 255:red, 0; green, 0; blue, 0 }  ,fill opacity=1 ] (200,220) -- (220,220) -- (220,240) -- (200,240) -- cycle ;
\draw  [fill={rgb, 255:red, 255; green, 255; blue, 255 }  ,fill opacity=1 ] (260,40) -- (280,40) -- (280,60) -- (260,60) -- cycle ;
\draw  [fill={rgb, 255:red, 0; green, 0; blue, 0 }  ,fill opacity=1 ] (260,60) -- (280,60) -- (280,80) -- (260,80) -- cycle ;
\draw   (260,80) -- (280,80) -- (280,100) -- (260,100) -- cycle ;
\draw  [fill={rgb, 255:red, 255; green, 255; blue, 255 }  ,fill opacity=1 ] (260,100) -- (280,100) -- (280,120) -- (260,120) -- cycle ;
\draw  [fill={rgb, 255:red, 255; green, 255; blue, 255 }  ,fill opacity=1 ] (260,120) -- (280,120) -- (280,140) -- (260,140) -- cycle ;
\draw  [fill={rgb, 255:red, 255; green, 255; blue, 255 }  ,fill opacity=1 ] (260,180) -- (280,180) -- (280,200) -- (260,200) -- cycle ;
\draw  [fill={rgb, 255:red, 0; green, 0; blue, 0 }  ,fill opacity=1 ] (268,150.5) .. controls (268,149.12) and (269.12,148) .. (270.5,148) .. controls (271.88,148) and (273,149.12) .. (273,150.5) .. controls (273,151.88) and (271.88,153) .. (270.5,153) .. controls (269.12,153) and (268,151.88) .. (268,150.5) -- cycle ;
\draw  [fill={rgb, 255:red, 0; green, 0; blue, 0 }  ,fill opacity=1 ] (268,160.5) .. controls (268,159.12) and (269.12,158) .. (270.5,158) .. controls (271.88,158) and (273,159.12) .. (273,160.5) .. controls (273,161.88) and (271.88,163) .. (270.5,163) .. controls (269.12,163) and (268,161.88) .. (268,160.5) -- cycle ;
\draw  [fill={rgb, 255:red, 0; green, 0; blue, 0 }  ,fill opacity=1 ] (268,170.5) .. controls (268,169.12) and (269.12,168) .. (270.5,168) .. controls (271.88,168) and (273,169.12) .. (273,170.5) .. controls (273,171.88) and (271.88,173) .. (270.5,173) .. controls (269.12,173) and (268,171.88) .. (268,170.5) -- cycle ;
\draw  [fill={rgb, 255:red, 255; green, 255; blue, 255 }  ,fill opacity=1 ] (260,200) -- (280,200) -- (280,220) -- (260,220) -- cycle ;
\draw  [fill={rgb, 255:red, 0; green, 0; blue, 0 }  ,fill opacity=1 ] (260,220) -- (280,220) -- (280,240) -- (260,240) -- cycle ;
\draw    (170,49.98) -- (200,50) ;
\draw    (120,49.98) -- (150,50) ;
\draw    (170,69.98) -- (200,70) ;
\draw    (120,69.98) -- (150,70) ;
\draw    (170,89.98) -- (200,90) ;
\draw    (120,89.98) -- (150,90) ;
\draw    (170,109.98) -- (200,110) ;
\draw    (120,109.98) -- (150,110) ;
\draw    (170,129.98) -- (200,130) ;
\draw    (120,129.98) -- (150,130) ;
\draw    (170,189.98) -- (200,190) ;
\draw    (120,189.98) -- (150,190) ;
\draw    (170,209.98) -- (200,210) ;
\draw    (120,209.98) -- (150,210) ;
\draw    (170,229.98) -- (200,230) ;
\draw    (120,229.98) -- (150,230) ;
\draw    (235,156.98) -- (245,157.07) ;
\draw    (235,161.98) -- (245,162.07) ;
\draw   (295,160) .. controls (295,157.24) and (297.24,155) .. (300,155) .. controls (302.76,155) and (305,157.24) .. (305,160) .. controls (305,162.76) and (302.76,165) .. (300,165) .. controls (297.24,165) and (295,162.76) .. (295,160) -- cycle ;
\draw    (300,155) -- (300,165) ;
\draw    (295,160) -- (305,160) ;
\draw    (280,89.8) -- (300,89.8) ;
\draw    (280,109.8) -- (300,109.8) ;
\draw    (280,129.8) -- (300,129.8) ;
\draw    (280,189.8) -- (300,189.8) ;
\draw    (280,209.8) -- (300,209.8) ;
\draw    (300,50.8) -- (300,150) ;
\draw    (300,170) -- (300,229.8) ;
\draw    (310,157.98) -- (320,158.07) ;
\draw    (310,162.98) -- (320,163.07) ;
\draw   (330,150) -- (350,150) -- (350,170) -- (330,170) -- cycle ;
\draw    (280,69.8) -- (300,69.8) ;
\draw    (280,50.8) -- (300,50.8) ;
\draw    (280,229.8) -- (300,229.8) ;

\draw (175,0) node [anchor=north west][inner sep=0.75pt]   [align=left] {input block};
\draw (203,18) node [anchor=north west][inner sep=0.75pt]    {$x_{j}$};
\draw (92,0) node [anchor=north west][inner sep=0.75pt]   [align=left] {index};
\draw (105,18) node [anchor=north west][inner sep=0.75pt]    {$i$};
\draw (155,45) node [anchor=north west][inner sep=0.75pt]  [font=\scriptsize] [align=left] {$\displaystyle \&$};
\draw (155,65) node [anchor=north west][inner sep=0.75pt]  [font=\scriptsize] [align=left] {$\displaystyle \&$};
\draw (155,105) node [anchor=north west][inner sep=0.75pt]  [font=\scriptsize] [align=left] {$\displaystyle \&$};
\draw (155,85) node [anchor=north west][inner sep=0.75pt]  [font=\scriptsize] [align=left] {$\displaystyle \&$};
\draw (155,125) node [anchor=north west][inner sep=0.75pt]  [font=\scriptsize] [align=left] {$\displaystyle \&$};
\draw (155,185) node [anchor=north west][inner sep=0.75pt]  [font=\scriptsize] [align=left] {$\displaystyle \&$};
\draw (155,205) node [anchor=north west][inner sep=0.75pt]  [font=\scriptsize] [align=left] {$\displaystyle \&$};
\draw (155,225) node [anchor=north west][inner sep=0.75pt]  [font=\scriptsize] [align=left] {$\displaystyle \&$};
\draw (150,-20) node [anchor=north west][inner sep=0.75pt]   [align=left] {\large\textbf{Construction 1}};

\end{tikzpicture}
}%
~~~~~~~~~~~~~~~~~~~~~~~~~~~~~
\resizebox{0.30\textwidth}{!}{%
\begin{tikzpicture}[x=0.75pt,y=0.75pt,yscale=-1,xscale=1]

\draw   (100,40) -- (120,40) -- (120,60) -- (100,60) -- cycle ;
\draw  [fill={rgb, 255:red, 0; green, 0; blue, 0 }  ,fill opacity=1 ] (100,60) -- (120,60) -- (120,80) -- (100,80) -- cycle ;
\draw   (100,80) -- (120,80) -- (120,100) -- (100,100) -- cycle ;
\draw   (100,100) -- (120,100) -- (120,120) -- (100,120) -- cycle ;
\draw  [fill={rgb, 255:red, 0; green, 0; blue, 0 }  ,fill opacity=1 ] (100,120) -- (120,120) -- (120,140) -- (100,140) -- cycle ;
\draw   (100,180) -- (120,180) -- (120,200) -- (100,200) -- cycle ;
\draw  [fill={rgb, 255:red, 0; green, 0; blue, 0 }  ,fill opacity=1 ] (108,150.5) .. controls (108,149.12) and (109.12,148) .. (110.5,148) .. controls (111.88,148) and (113,149.12) .. (113,150.5) .. controls (113,151.88) and (111.88,153) .. (110.5,153) .. controls (109.12,153) and (108,151.88) .. (108,150.5) -- cycle ;
\draw  [fill={rgb, 255:red, 0; green, 0; blue, 0 }  ,fill opacity=1 ] (108,160.5) .. controls (108,159.12) and (109.12,158) .. (110.5,158) .. controls (111.88,158) and (113,159.12) .. (113,160.5) .. controls (113,161.88) and (111.88,163) .. (110.5,163) .. controls (109.12,163) and (108,161.88) .. (108,160.5) -- cycle ;
\draw  [fill={rgb, 255:red, 0; green, 0; blue, 0 }  ,fill opacity=1 ] (108,170.5) .. controls (108,169.12) and (109.12,168) .. (110.5,168) .. controls (111.88,168) and (113,169.12) .. (113,170.5) .. controls (113,171.88) and (111.88,173) .. (110.5,173) .. controls (109.12,173) and (108,171.88) .. (108,170.5) -- cycle ;
\draw  [fill={rgb, 255:red, 0; green, 0; blue, 0 }  ,fill opacity=1 ] (100,200) -- (120,200) -- (120,220) -- (100,220) -- cycle ;
\draw  [fill={rgb, 255:red, 0; green, 0; blue, 0 }  ,fill opacity=1 ] (100,220) -- (120,220) -- (120,240) -- (100,240) -- cycle ;
\draw  [fill={rgb, 255:red, 0; green, 0; blue, 0 }  ,fill opacity=1 ] (200,40) -- (220,40) -- (220,60) -- (200,60) -- cycle ;
\draw  [fill={rgb, 255:red, 0; green, 0; blue, 0 }  ,fill opacity=1 ] (200,60) -- (220,60) -- (220,80) -- (200,80) -- cycle ;
\draw   (200,80) -- (220,80) -- (220,100) -- (200,100) -- cycle ;
\draw  [fill={rgb, 255:red, 0; green, 0; blue, 0 }  ,fill opacity=1 ] (200,100) -- (220,100) -- (220,120) -- (200,120) -- cycle ;
\draw   (200,120) -- (220,120) -- (220,140) -- (200,140) -- cycle ;
\draw  [fill={rgb, 255:red, 0; green, 0; blue, 0 }  ,fill opacity=1 ] (200,180) -- (220,180) -- (220,200) -- (200,200) -- cycle ;
\draw  [fill={rgb, 255:red, 0; green, 0; blue, 0 }  ,fill opacity=1 ] (208,150.5) .. controls (208,149.12) and (209.12,148) .. (210.5,148) .. controls (211.88,148) and (213,149.12) .. (213,150.5) .. controls (213,151.88) and (211.88,153) .. (210.5,153) .. controls (209.12,153) and (208,151.88) .. (208,150.5) -- cycle ;
\draw  [fill={rgb, 255:red, 0; green, 0; blue, 0 }  ,fill opacity=1 ] (208,160.5) .. controls (208,159.12) and (209.12,158) .. (210.5,158) .. controls (211.88,158) and (213,159.12) .. (213,160.5) .. controls (213,161.88) and (211.88,163) .. (210.5,163) .. controls (209.12,163) and (208,161.88) .. (208,160.5) -- cycle ;
\draw  [fill={rgb, 255:red, 0; green, 0; blue, 0 }  ,fill opacity=1 ] (208,170.5) .. controls (208,169.12) and (209.12,168) .. (210.5,168) .. controls (211.88,168) and (213,169.12) .. (213,170.5) .. controls (213,171.88) and (211.88,173) .. (210.5,173) .. controls (209.12,173) and (208,171.88) .. (208,170.5) -- cycle ;
\draw   (200,200) -- (220,200) -- (220,220) -- (200,220) -- cycle ;
\draw  [fill={rgb, 255:red, 0; green, 0; blue, 0 }  ,fill opacity=1 ] (200,220) -- (220,220) -- (220,240) -- (200,240) -- cycle ;
\draw  [fill={rgb, 255:red, 0; green, 0; blue, 0 }  ,fill opacity=1 ] (260,40) -- (280,40) -- (280,60) -- (260,60) -- cycle ;
\draw   (260,60) -- (280,60) -- (280,80) -- (260,80) -- cycle ;
\draw   (260,80) -- (280,80) -- (280,100) -- (260,100) -- cycle ;
\draw  [fill={rgb, 255:red, 0; green, 0; blue, 0 }  ,fill opacity=1 ] (260,100) -- (280,100) -- (280,120) -- (260,120) -- cycle ;
\draw  [fill={rgb, 255:red, 0; green, 0; blue, 0 }  ,fill opacity=1 ] (260,120) -- (280,120) -- (280,140) -- (260,140) -- cycle ;
\draw  [fill={rgb, 255:red, 0; green, 0; blue, 0 }  ,fill opacity=1 ] (260,180) -- (280,180) -- (280,200) -- (260,200) -- cycle ;
\draw  [fill={rgb, 255:red, 0; green, 0; blue, 0 }  ,fill opacity=1 ] (268,150.5) .. controls (268,149.12) and (269.12,148) .. (270.5,148) .. controls (271.88,148) and (273,149.12) .. (273,150.5) .. controls (273,151.88) and (271.88,153) .. (270.5,153) .. controls (269.12,153) and (268,151.88) .. (268,150.5) -- cycle ;
\draw  [fill={rgb, 255:red, 0; green, 0; blue, 0 }  ,fill opacity=1 ] (268,160.5) .. controls (268,159.12) and (269.12,158) .. (270.5,158) .. controls (271.88,158) and (273,159.12) .. (273,160.5) .. controls (273,161.88) and (271.88,163) .. (270.5,163) .. controls (269.12,163) and (268,161.88) .. (268,160.5) -- cycle ;
\draw  [fill={rgb, 255:red, 0; green, 0; blue, 0 }  ,fill opacity=1 ] (268,170.5) .. controls (268,169.12) and (269.12,168) .. (270.5,168) .. controls (271.88,168) and (273,169.12) .. (273,170.5) .. controls (273,171.88) and (271.88,173) .. (270.5,173) .. controls (269.12,173) and (268,171.88) .. (268,170.5) -- cycle ;
\draw  [fill={rgb, 255:red, 0; green, 0; blue, 0 }  ,fill opacity=1 ] (260,200) -- (280,200) -- (280,220) -- (260,220) -- cycle ;
\draw   (260,220) -- (280,220) -- (280,240) -- (260,240) -- cycle ;
\draw   (155,50) .. controls (155,47.24) and (157.24,45) .. (160,45) .. controls (162.76,45) and (165,47.24) .. (165,50) .. controls (165,52.76) and (162.76,55) .. (160,55) .. controls (157.24,55) and (155,52.76) .. (155,50) -- cycle ;
\draw   (155,70) .. controls (155,67.24) and (157.24,65) .. (160,65) .. controls (162.76,65) and (165,67.24) .. (165,70) .. controls (165,72.76) and (162.76,75) .. (160,75) .. controls (157.24,75) and (155,72.76) .. (155,70) -- cycle ;
\draw   (155,90) .. controls (155,87.24) and (157.24,85) .. (160,85) .. controls (162.76,85) and (165,87.24) .. (165,90) .. controls (165,92.76) and (162.76,95) .. (160,95) .. controls (157.24,95) and (155,92.76) .. (155,90) -- cycle ;
\draw   (155,110) .. controls (155,107.24) and (157.24,105) .. (160,105) .. controls (162.76,105) and (165,107.24) .. (165,110) .. controls (165,112.76) and (162.76,115) .. (160,115) .. controls (157.24,115) and (155,112.76) .. (155,110) -- cycle ;
\draw   (155,130) .. controls (155,127.24) and (157.24,125) .. (160,125) .. controls (162.76,125) and (165,127.24) .. (165,130) .. controls (165,132.76) and (162.76,135) .. (160,135) .. controls (157.24,135) and (155,132.76) .. (155,130) -- cycle ;
\draw   (155,190) .. controls (155,187.24) and (157.24,185) .. (160,185) .. controls (162.76,185) and (165,187.24) .. (165,190) .. controls (165,192.76) and (162.76,195) .. (160,195) .. controls (157.24,195) and (155,192.76) .. (155,190) -- cycle ;
\draw   (155,210) .. controls (155,207.24) and (157.24,205) .. (160,205) .. controls (162.76,205) and (165,207.24) .. (165,210) .. controls (165,212.76) and (162.76,215) .. (160,215) .. controls (157.24,215) and (155,212.76) .. (155,210) -- cycle ;
\draw   (155,230) .. controls (155,227.24) and (157.24,225) .. (160,225) .. controls (162.76,225) and (165,227.24) .. (165,230) .. controls (165,232.76) and (162.76,235) .. (160,235) .. controls (157.24,235) and (155,232.76) .. (155,230) -- cycle ;
\draw    (160,125) -- (160,135) ;
\draw    (155,130) -- (165,130) ;
\draw    (160,105) -- (160,115) ;
\draw    (155,110) -- (165,110) ;
\draw    (160,185) -- (160,195) ;
\draw    (155,190) -- (165,190) ;
\draw    (160,205) -- (160,215) ;
\draw    (155,210) -- (165,210) ;
\draw    (160,225) -- (160,235) ;
\draw    (155,230) -- (165,230) ;
\draw    (160,45) -- (160,55) ;
\draw    (155,50) -- (165,50) ;
\draw    (160,65) -- (160,75) ;
\draw    (155,70) -- (165,70) ;
\draw    (160,85) -- (160,95) ;
\draw    (155,90) -- (165,90) ;
\draw    (170,49.98) -- (200,50) ;
\draw    (120,49.98) -- (150,50) ;
\draw    (170,69.98) -- (200,70) ;
\draw    (120,69.98) -- (150,70) ;
\draw    (170,89.98) -- (200,90) ;
\draw    (120,89.98) -- (150,90) ;
\draw    (170,109.98) -- (200,110) ;
\draw    (120,109.98) -- (150,110) ;
\draw    (170,129.98) -- (200,130) ;
\draw    (120,129.98) -- (150,130) ;
\draw    (170,189.98) -- (200,190) ;
\draw    (120,189.98) -- (150,190) ;
\draw    (170,209.98) -- (200,210) ;
\draw    (120,209.98) -- (150,210) ;
\draw    (170,229.98) -- (200,230) ;
\draw    (120,229.98) -- (150,230) ;
\draw    (235,156.98) -- (245,157.07) ;
\draw    (235,161.98) -- (245,162.07) ;
\draw    (280,49.98) -- (300,49.98) ;
\draw    (300,49.98) ;
\draw    (300,49.98) -- (300,55) ;
\draw    (280,69.98) -- (300,69.98) ;
\draw    (300,64.97) -- (300,69.98) ;
\draw    (280,89.98) -- (300,89.98) ;
\draw    (300,89.98) ;
\draw    (300,89.98) -- (300,95) ;
\draw    (280,109.98) -- (300,109.98) ;
\draw    (300,104.97) -- (300,109.98) ;
\draw    (280,129.98) -- (300,129.98) ;
\draw    (300,129.98) ;
\draw    (300,129.98) -- (300,135) ;
\draw    (280,149.98) -- (300,149.98) ;
\draw    (300,144.97) -- (300,149.98) ;
\draw    (280,169.98) -- (300,169.98) ;
\draw    (300,169.98) ;
\draw    (300,169.98) -- (300,175) ;
\draw    (280,189.98) -- (300,189.98) ;
\draw    (300,184.97) -- (300,189.98) ;
\draw    (280,210.38) -- (300,210.38) ;
\draw    (300,210.38) ;
\draw    (300,210.38) -- (300,215.4) ;
\draw    (280,230.38) -- (300,230.38) ;
\draw    (300,225.37) -- (300,230.38) ;
\draw    (315,157.98) -- (325,158.07) ;
\draw    (315,162.98) -- (325,163.07) ;
\draw   (340,80) -- (360,80) -- (360,100) -- (340,100) -- cycle ;
\draw   (340,100) -- (360,100) -- (360,120) -- (340,120) -- cycle ;
\draw  [fill={rgb, 255:red, 0; green, 0; blue, 0 }  ,fill opacity=1 ] (340,120) -- (360,120) -- (360,140) -- (340,140) -- cycle ;
\draw  [fill={rgb, 255:red, 0; green, 0; blue, 0 }  ,fill opacity=1 ] (348,150.5) .. controls (348,149.12) and (349.12,148) .. (350.5,148) .. controls (351.88,148) and (353,149.12) .. (353,150.5) .. controls (353,151.88) and (351.88,153) .. (350.5,153) .. controls (349.12,153) and (348,151.88) .. (348,150.5) -- cycle ;
\draw  [fill={rgb, 255:red, 0; green, 0; blue, 0 }  ,fill opacity=1 ] (348,160.5) .. controls (348,159.12) and (349.12,158) .. (350.5,158) .. controls (351.88,158) and (353,159.12) .. (353,160.5) .. controls (353,161.88) and (351.88,163) .. (350.5,163) .. controls (349.12,163) and (348,161.88) .. (348,160.5) -- cycle ;
\draw  [fill={rgb, 255:red, 0; green, 0; blue, 0 }  ,fill opacity=1 ] (348,170.5) .. controls (348,169.12) and (349.12,168) .. (350.5,168) .. controls (351.88,168) and (353,169.12) .. (353,170.5) .. controls (353,171.88) and (351.88,173) .. (350.5,173) .. controls (349.12,173) and (348,171.88) .. (348,170.5) -- cycle ;
\draw  [fill={rgb, 255:red, 0; green, 0; blue, 0 }  ,fill opacity=1 ] (340,180) -- (360,180) -- (360,200) -- (340,200) -- cycle ;
\draw   (340,200) -- (360,200) -- (360,220) -- (340,220) -- cycle ;
\draw   (375,160) .. controls (375,157.24) and (377.24,155) .. (380,155) .. controls (382.76,155) and (385,157.24) .. (385,160) .. controls (385,162.76) and (382.76,165) .. (380,165) .. controls (377.24,165) and (375,162.76) .. (375,160) -- cycle ;
\draw    (380,155) -- (380,165) ;
\draw    (375,160) -- (385,160) ;
\draw    (360,89.8) -- (380,89.8) ;
\draw    (360,109.8) -- (380,109.8) ;
\draw    (360,129.8) -- (380,129.8) ;
\draw    (360,189.8) -- (380,189.8) ;
\draw    (360,209.8) -- (380,209.8) ;
\draw    (380,89.8) -- (380,150) ;
\draw    (380,170) -- (380,209.8) ;
\draw    (390,157.98) -- (400,158.07) ;
\draw    (390,162.98) -- (400,163.07) ;
\draw   (410,150) -- (430,150) -- (430,170) -- (410,170) -- cycle ;

\draw (295,55) node [anchor=north west][inner sep=0.75pt]  [font=\scriptsize] [align=left] {$\displaystyle \&$};
\draw (295,95) node [anchor=north west][inner sep=0.75pt]  [font=\scriptsize] [align=left] {$\displaystyle \&$};
\draw (295,135) node [anchor=north west][inner sep=0.75pt]  [font=\scriptsize] [align=left] {$\displaystyle \&$};
\draw (295,175) node [anchor=north west][inner sep=0.75pt]  [font=\scriptsize] [align=left] {$\displaystyle \&$};
\draw (295,215) node [anchor=north west][inner sep=0.75pt]  [font=\scriptsize] [align=left] {$\displaystyle \&$};
\draw (175,0) node [anchor=north west][inner sep=0.75pt]   [align=left] {input block};
\draw (203,18) node [anchor=north west][inner sep=0.75pt]    {$x_{j}$};
\draw (92,0) node [anchor=north west][inner sep=0.75pt]   [align=left] {index};
\draw (105,18) node [anchor=north west][inner sep=0.75pt]    {$i$};
\draw (175,-20) node [anchor=north west][inner sep=0.75pt]   [align=left] {\large\textbf{Construction 2}};

\end{tikzpicture}
}%
~~~~~~~~~~~~~~~~~~~~~~~~~~~~~
%
\resizebox{0.31\textwidth}{!}{%
\begin{tikzpicture}[x=0.75pt,y=0.75pt,yscale=-1,xscale=1]

\draw  [fill={rgb, 255:red, 0; green, 0; blue, 0 }  ,fill opacity=1 ] (100,40) -- (120,40) -- (120,60) -- (100,60) -- cycle ;
\draw   (100,60) -- (120,60) -- (120,80) -- (100,80) -- cycle ;
\draw   (100,80) -- (120,80) -- (120,100) -- (100,100) -- cycle ;
\draw  [fill={rgb, 255:red, 0; green, 0; blue, 0 }  ,fill opacity=1 ] (100,100) -- (120,100) -- (120,120) -- (100,120) -- cycle ;
\draw  [fill={rgb, 255:red, 0; green, 0; blue, 0 }  ,fill opacity=1 ] (100,120) -- (120,120) -- (120,140) -- (100,140) -- cycle ;
\draw   (100,180) -- (120,180) -- (120,200) -- (100,200) -- cycle ;
\draw  [fill={rgb, 255:red, 0; green, 0; blue, 0 }  ,fill opacity=1 ] (108,150.5) .. controls (108,149.12) and (109.12,148) .. (110.5,148) .. controls (111.88,148) and (113,149.12) .. (113,150.5) .. controls (113,151.88) and (111.88,153) .. (110.5,153) .. controls (109.12,153) and (108,151.88) .. (108,150.5) -- cycle ;
\draw  [fill={rgb, 255:red, 0; green, 0; blue, 0 }  ,fill opacity=1 ] (108,160.5) .. controls (108,159.12) and (109.12,158) .. (110.5,158) .. controls (111.88,158) and (113,159.12) .. (113,160.5) .. controls (113,161.88) and (111.88,163) .. (110.5,163) .. controls (109.12,163) and (108,161.88) .. (108,160.5) -- cycle ;
\draw  [fill={rgb, 255:red, 0; green, 0; blue, 0 }  ,fill opacity=1 ] (108,170.5) .. controls (108,169.12) and (109.12,168) .. (110.5,168) .. controls (111.88,168) and (113,169.12) .. (113,170.5) .. controls (113,171.88) and (111.88,173) .. (110.5,173) .. controls (109.12,173) and (108,171.88) .. (108,170.5) -- cycle ;
\draw  (100,200) -- (120,200) -- (120,220) -- (100,220) -- cycle ;
\draw  [fill={rgb, 255:red, 0; green, 0; blue, 0 }  ,fill opacity=1 ] (100,220) -- (120,220) -- (120,240) -- (100,240) -- cycle ;
\draw  [fill={rgb, 255:red, 0; green, 0; blue, 0 }  ,fill opacity=1 ] (200,40) -- (220,40) -- (220,60) -- (200,60) -- cycle ;
\draw  [fill={rgb, 255:red, 0; green, 0; blue, 0 }  ,fill opacity=1 ] (200,60) -- (220,60) -- (220,80) -- (200,80) -- cycle ;
\draw   (200,80) -- (220,80) -- (220,100) -- (200,100) -- cycle ;
\draw  [fill={rgb, 255:red, 0; green, 0; blue, 0 }  ,fill opacity=1 ] (200,100) -- (220,100) -- (220,120) -- (200,120) -- cycle ;
\draw   (200,120) -- (220,120) -- (220,140) -- (200,140) -- cycle ;
\draw  [fill={rgb, 255:red, 0; green, 0; blue, 0 }  ,fill opacity=1 ] (200,180) -- (220,180) -- (220,200) -- (200,200) -- cycle ;
\draw  [fill={rgb, 255:red, 0; green, 0; blue, 0 }  ,fill opacity=1 ] (208,150.5) .. controls (208,149.12) and (209.12,148) .. (210.5,148) .. controls (211.88,148) and (213,149.12) .. (213,150.5) .. controls (213,151.88) and (211.88,153) .. (210.5,153) .. controls (209.12,153) and (208,151.88) .. (208,150.5) -- cycle ;
\draw  [fill={rgb, 255:red, 0; green, 0; blue, 0 }  ,fill opacity=1 ] (208,160.5) .. controls (208,159.12) and (209.12,158) .. (210.5,158) .. controls (211.88,158) and (213,159.12) .. (213,160.5) .. controls (213,161.88) and (211.88,163) .. (210.5,163) .. controls (209.12,163) and (208,161.88) .. (208,160.5) -- cycle ;
\draw  [fill={rgb, 255:red, 0; green, 0; blue, 0 }  ,fill opacity=1 ] (208,170.5) .. controls (208,169.12) and (209.12,168) .. (210.5,168) .. controls (211.88,168) and (213,169.12) .. (213,170.5) .. controls (213,171.88) and (211.88,173) .. (210.5,173) .. controls (209.12,173) and (208,171.88) .. (208,170.5) -- cycle ;
\draw   (200,200) -- (220,200) -- (220,220) -- (200,220) -- cycle ;
\draw  [fill={rgb, 255:red, 0; green, 0; blue, 0 }  ,fill opacity=1 ] (200,220) -- (220,220) -- (220,240) -- (200,240) -- cycle ;
\draw   (260,40) -- (280,40) -- (280,60) -- (260,60) -- cycle ;
\draw   [fill={rgb, 255:red, 0; green, 0; blue, 0 }  ,fill opacity=1 ] (260,60) -- (280,60) -- (280,80) -- (260,80) -- cycle ;
\draw   (260,80) -- (280,80) -- (280,100) -- (260,100) -- cycle ;
\draw   (260,100) -- (280,100) -- (280,120) -- (260,120) -- cycle ;
\draw  [fill={rgb, 255:red, 0; green, 0; blue, 0 }  ,fill opacity=1 ] (260,120) -- (280,120) -- (280,140) -- (260,140) -- cycle ;
\draw  [fill={rgb, 255:red, 0; green, 0; blue, 0 }  ,fill opacity=1 ] (260,180) -- (280,180) -- (280,200) -- (260,200) -- cycle ;
\draw  [fill={rgb, 255:red, 0; green, 0; blue, 0 }  ,fill opacity=1 ] (268,150.5) .. controls (268,149.12) and (269.12,148) .. (270.5,148) .. controls (271.88,148) and (273,149.12) .. (273,150.5) .. controls (273,151.88) and (271.88,153) .. (270.5,153) .. controls (269.12,153) and (268,151.88) .. (268,150.5) -- cycle ;
\draw  [fill={rgb, 255:red, 0; green, 0; blue, 0 }  ,fill opacity=1 ] (268,160.5) .. controls (268,159.12) and (269.12,158) .. (270.5,158) .. controls (271.88,158) and (273,159.12) .. (273,160.5) .. controls (273,161.88) and (271.88,163) .. (270.5,163) .. controls (269.12,163) and (268,161.88) .. (268,160.5) -- cycle ;
\draw  [fill={rgb, 255:red, 0; green, 0; blue, 0 }  ,fill opacity=1 ] (268,170.5) .. controls (268,169.12) and (269.12,168) .. (270.5,168) .. controls (271.88,168) and (273,169.12) .. (273,170.5) .. controls (273,171.88) and (271.88,173) .. (270.5,173) .. controls (269.12,173) and (268,171.88) .. (268,170.5) -- cycle ;
\draw   (260,200) -- (280,200) -- (280,220) -- (260,220) -- cycle ;
\draw   (260,220) -- (280,220) -- (280,240) -- (260,240) -- cycle ;
\draw   (155,50) .. controls (155,47.24) and (157.24,45) .. (160,45) .. controls (162.76,45) and (165,47.24) .. (165,50) .. controls (165,52.76) and (162.76,55) .. (160,55) .. controls (157.24,55) and (155,52.76) .. (155,50) -- cycle ;
\draw   (155,70) .. controls (155,67.24) and (157.24,65) .. (160,65) .. controls (162.76,65) and (165,67.24) .. (165,70) .. controls (165,72.76) and (162.76,75) .. (160,75) .. controls (157.24,75) and (155,72.76) .. (155,70) -- cycle ;
\draw   (155,90) .. controls (155,87.24) and (157.24,85) .. (160,85) .. controls (162.76,85) and (165,87.24) .. (165,90) .. controls (165,92.76) and (162.76,95) .. (160,95) .. controls (157.24,95) and (155,92.76) .. (155,90) -- cycle ;
\draw   (155,110) .. controls (155,107.24) and (157.24,105) .. (160,105) .. controls (162.76,105) and (165,107.24) .. (165,110) .. controls (165,112.76) and (162.76,115) .. (160,115) .. controls (157.24,115) and (155,112.76) .. (155,110) -- cycle ;
\draw   (155,130) .. controls (155,127.24) and (157.24,125) .. (160,125) .. controls (162.76,125) and (165,127.24) .. (165,130) .. controls (165,132.76) and (162.76,135) .. (160,135) .. controls (157.24,135) and (155,132.76) .. (155,130) -- cycle ;
\draw   (155,190) .. controls (155,187.24) and (157.24,185) .. (160,185) .. controls (162.76,185) and (165,187.24) .. (165,190) .. controls (165,192.76) and (162.76,195) .. (160,195) .. controls (157.24,195) and (155,192.76) .. (155,190) -- cycle ;
\draw   (155,210) .. controls (155,207.24) and (157.24,205) .. (160,205) .. controls (162.76,205) and (165,207.24) .. (165,210) .. controls (165,212.76) and (162.76,215) .. (160,215) .. controls (157.24,215) and (155,212.76) .. (155,210) -- cycle ;
\draw   (155,230) .. controls (155,227.24) and (157.24,225) .. (160,225) .. controls (162.76,225) and (165,227.24) .. (165,230) .. controls (165,232.76) and (162.76,235) .. (160,235) .. controls (157.24,235) and (155,232.76) .. (155,230) -- cycle ;
\draw    (160,125) -- (160,135) ;
\draw    (155,130) -- (165,130) ;
\draw    (160,105) -- (160,115) ;
\draw    (155,110) -- (165,110) ;
\draw    (160,185) -- (160,195) ;
\draw    (155,190) -- (165,190) ;
\draw    (160,205) -- (160,215) ;
\draw    (155,210) -- (165,210) ;
\draw    (160,225) -- (160,235) ;
\draw    (155,230) -- (165,230) ;
\draw    (160,45) -- (160,55) ;
\draw    (155,50) -- (165,50) ;
\draw    (160,65) -- (160,75) ;
\draw    (155,70) -- (165,70) ;
\draw    (160,85) -- (160,95) ;
\draw    (155,90) -- (165,90) ;
\draw    (170,49.98) -- (200,50) ;
\draw    (120,49.98) -- (150,50) ;
\draw    (170,69.98) -- (200,70) ;
\draw    (120,69.98) -- (150,70) ;
\draw    (170,89.98) -- (200,90) ;
\draw    (120,89.98) -- (150,90) ;
\draw    (170,109.98) -- (200,110) ;
\draw    (120,109.98) -- (150,110) ;
\draw    (170,129.98) -- (200,130) ;
\draw    (120,129.98) -- (150,130) ;
\draw    (170,189.98) -- (200,190) ;
\draw    (120,189.98) -- (150,190) ;
\draw    (170,209.98) -- (200,210) ;
\draw    (120,209.98) -- (150,210) ;
\draw    (170,229.98) -- (200,230) ;
\draw    (120,229.98) -- (150,230) ;
\draw    (235,156.98) -- (245,157.07) ;
\draw    (235,161.98) -- (245,162.07) ;
\draw    (280,49.98) -- (300,49.98) ;
\draw    (300,49.98) ;
\draw    (300,49.98) -- (300,55) ;
\draw    (280,69.98) -- (300,69.98) ;
\draw    (300,64.97) -- (300,69.98) ;
\draw    (280,89.98) -- (300,89.98) ;
\draw    (300,89.98) ;
\draw    (300,89.98) -- (300,95) ;
\draw    (280,109.98) -- (300,109.98) ;
\draw    (300,104.97) -- (300,109.98) ;
\draw    (280,129.98) -- (300,129.98) ;
\draw    (300,129.98) ;
\draw    (300,129.98) -- (300,135) ;
\draw    (280,149.98) -- (300,149.98) ;
\draw    (300,144.97) -- (300,149.98) ;
\draw    (280,169.98) -- (300,169.98) ;
\draw    (300,169.98) ;
\draw    (300,169.98) -- (300,175) ;
\draw    (280,189.98) -- (300,189.98) ;
\draw    (300,184.97) -- (300,189.98) ;
\draw    (280,210.38) -- (300,210.38) ;
\draw    (300,210.38) ;
\draw    (300,210.38) -- (300,215.4) ;
\draw    (280,230.38) -- (300,230.38) ;
\draw    (300,225.37) -- (300,230.38) ;
\draw    (315,157.98) -- (325,158.07) ;
\draw    (315,162.98) -- (325,163.07) ;
\draw   (340,80) -- (360,80) -- (360,100) -- (340,100) -- cycle ;
\draw   (340,100) -- (360,100) -- (360,120) -- (340,120) -- cycle ;
\draw  [fill={rgb, 255:red, 0; green, 0; blue, 0 }  ,fill opacity=1 ] (340,120) -- (360,120) -- (360,140) -- (340,140) -- cycle ;
\draw  [fill={rgb, 255:red, 0; green, 0; blue, 0 }  ,fill opacity=1 ] (348,150.5) .. controls (348,149.12) and (349.12,148) .. (350.5,148) .. controls (351.88,148) and (353,149.12) .. (353,150.5) .. controls (353,151.88) and (351.88,153) .. (350.5,153) .. controls (349.12,153) and (348,151.88) .. (348,150.5) -- cycle ;
\draw  [fill={rgb, 255:red, 0; green, 0; blue, 0 }  ,fill opacity=1 ] (348,160.5) .. controls (348,159.12) and (349.12,158) .. (350.5,158) .. controls (351.88,158) and (353,159.12) .. (353,160.5) .. controls (353,161.88) and (351.88,163) .. (350.5,163) .. controls (349.12,163) and (348,161.88) .. (348,160.5) -- cycle ;
\draw  [fill={rgb, 255:red, 0; green, 0; blue, 0 }  ,fill opacity=1 ] (348,170.5) .. controls (348,169.12) and (349.12,168) .. (350.5,168) .. controls (351.88,168) and (353,169.12) .. (353,170.5) .. controls (353,171.88) and (351.88,173) .. (350.5,173) .. controls (349.12,173) and (348,171.88) .. (348,170.5) -- cycle ;
\draw  [fill={rgb, 255:red, 0; green, 0; blue, 0 }  ,fill opacity=1 ] (340,180) -- (360,180) -- (360,200) -- (340,200) -- cycle ;
\draw   (340,200) -- (360,200) -- (360,220) -- (340,220) -- cycle ;
\draw   (375,160) .. controls (375,157.24) and (377.24,155) .. (380,155) .. controls (382.76,155) and (385,157.24) .. (385,160) .. controls (385,162.76) and (382.76,165) .. (380,165) .. controls (377.24,165) and (375,162.76) .. (375,160) -- cycle ;
\draw    (380,155) -- (380,165) ;
\draw    (375,160) -- (385,160) ;
\draw    (360,89.8) -- (380,89.8) ;
\draw    (360,109.8) -- (380,109.8) ;
\draw    (360,129.8) -- (380,129.8) ;
\draw    (360,189.8) -- (380,189.8) ;
\draw    (360,209.8) -- (380,209.8) ;
\draw    (380,89.8) -- (380,150) ;
\draw    (380,170) -- (380,209.8) ;
\draw    (390,157.98) -- (400,158.07) ;
\draw    (390,162.98) -- (400,163.07) ;
\draw   (410,150) -- (430,150) -- (430,170) -- (410,170) -- cycle ;

\draw (295,55) node [anchor=north west][inner sep=0.75pt]  [font=\scriptsize] [align=left] {$\displaystyle \&$};
\draw (295,95) node [anchor=north west][inner sep=0.75pt]  [font=\scriptsize] [align=left] {$\displaystyle \&$};
\draw (295,135) node [anchor=north west][inner sep=0.75pt]  [font=\scriptsize] [align=left] {$\displaystyle \&$};
\draw (295,175) node [anchor=north west][inner sep=0.75pt]  [font=\scriptsize] [align=left] {$\displaystyle \&$};
\draw (295,215) node [anchor=north west][inner sep=0.75pt]  [font=\scriptsize] [align=left] {$\displaystyle \&$};
\draw (175,0) node [anchor=north west][inner sep=0.75pt]   [align=left] {input block};
\draw (203,18) node [anchor=north west][inner sep=0.75pt]    {$x_{j}$};
\draw (85,0) node [anchor=north west][inner sep=0.75pt]   [align=left] {random};
\draw (82,16) node [anchor=north west][inner sep=0.75pt]    {$\mathtt{prng}( i,b)$};
\draw (175,-20) node [anchor=north west][inner sep=0.75pt]   [align=left] {\large\textbf{Construction 3}};

\end{tikzpicture}
}%

\caption{A depiction of digest functions $d_1$ (left), $d_2$ (center), and $d_3$ (right).}
\label{fig:constructions-1-2-3}
\end{figure*}

\begin{theorem}
$\mathcal{C}_1$ preserves $1/2$-collision boundedness.
\end{theorem}
\begin{proof}
If $x$ and $x'$ differ, they differ in some block $j$. Let $\process_{b, d}(i,x)_j$ be the $j$th bit of $\process_{b, d_1}(i,x)$.  Note that $\process_{b, d_1}(i,x)_j=\process_{b, d_1}(i,x)_j$ if and only if $d_1(i, x_j) = d_1(i, x_j)$, which happens if and only if $\parity((x_j \xor x'_j) \& i)=0$ (as conjunction distributes over exclusive or). As $x_j \xor x'_j\neq 0$, we have that $\process_{b, d_1}(i,x)=\process_{b, d_1}(i,x)$ holds for exactly half of the possible values of $i$. The result is then immediate from Theorem~\ref{thm:collision}.
\end{proof}

Note that $d_1$ (and thus $\process_{b, d_1}(i,x)$) has only one AND gate per bit of input and the output of $\process_{b, d_1}(i,x)$ has size $\lceil n/b \rceil$ (see Figure~\ref{fig:constructions-1-2-3}). Thus for large $n$ and $b$ this construction can asymptotically be computed with a number of AND gates arbitrarily close to $n$ (and $n$ XOR gates). However, as we will see in the lower bounds section (Section~\ref{sec:lower}) only $n/2$ AND gates are needed for any collision resistant indexed hash function. This motivates looking for the following construction which closes this gap.

\myparagraph{Construction $2$: $\frac{n}{2}$ AND gates.}
We wish to avoid using an AND gate for each bit of the input but still need some nonlinearity in $d$. So the idea is to combine two bits of the input together using a single AND gate. If the raw input bits went directly into the AND gate then the adversary would sometimes be able to change them in a way that definitely would not change the output of the gate. Instead we will first XOR each bit with a bit from $i$. Now the adversary cannot tell whether changing a certain bit will change the output.

Let $\xor$ denote bitwise XOR. Let $y$ be a bitstring of length $2m$ and let $y_j$ be the $j$th bit of $y$, we define $\andreduce(y)$, as the concatenation of $y_{2j} \& y_{2j+1}$ for all $j \in (0,...,m-1)$. Then we define $d_2(i,x_j)$ to be $\parity(\andreduce(x_j\xor i))$, see Figure~\ref{fig:constructions-1-2-3} (center).

A \emph{bent function} has the property that for a fixed linear change to its input, the output of the function would change for exactly half of all starting inputs~\cite{ROTHAUS1976bent}. The following theorem boils down to showing that $\parity \circ \andreduce$ is a bent function.

\begin{theorem}\label{thm:construction2}
$\mathcal{C}_2$ preserves $1/2$-collision boundedness.
\end{theorem}
\begin{proof}
By Theorem~\ref{thm:collision} it suffices to show that if $x_j\neq x'_j$, i.e. $x$ and $x'$ differ in the $j$th block, then $e_j := d_2(x_j) \xor d_2(x_j') = \parity(\andreduce(x_j\xor i)) \xor \parity(\andreduce(x'_j\xor i))$ is a uniformly random bit for a randomly chosen $i\in \mathcal{I}$.

Given that $x_j\neq x'_j$ we assume WLOG that they differ in at least one of the first two bits. The first bit of $\andreduce(x_j\xor i)$ is $1$ if and only if the first two bits of $i$ are the bitwise not of the first two bits of $x_j$, therefore it is $1$ with probability $1/4$. Similarly the first bit of $\andreduce(x'_j\xor i)$ is $1$ with probability $1/4$ and they can not both be $1$ at once. Therefore they differ with probability $1/2$. Further they are independent of all but the first two bits of $i$ and thus of the rest of $\andreduce(x_j\xor i)$ and $\andreduce(x'_j\xor i)$. It follows that $e_j$ is uniformly random.
\end{proof}

This construction can approach arbitrarily close to half an AND gate per bit of input by
choosing $b$ appropriately (see Table~\ref{tab:asymptotics-small}).
However, suppose $h$ is SHA3-256, which requires $\approx\! 35$ AND gates per bit of input (and thus $C_N \!\approx \!35$ in Table~\ref{tab:asymptotics-small}). Then for $H$ to have less than one AND gate per bit we would have to take $b\!\geq\! 70$. This would give $|\mathcal{I}|\!\geq\! 2^{70}$.
Recall that $|\mathcal{I}|$ directly corresponds to the size of our commitments, and thus $|\mathcal{I}|\!=\!2^b$ is impractical in computation and communication/storage.
Construction $3$ sacrifices a small amount of collision resistance in order to reduce the size of $\mathcal{I}$.

\myparagraph{Construction $3$: main result.}
The digest function $d_3$ is shown in Figure~\ref{fig:constructions-1-2-3} (right), and is analogous to $d_2$, but instead of using $i\in \mathcal{I}=\{0,1\}^b$ directly, we
use a pseudo-random number generator $\prng$ to expand $i$ into length $b$ strings to xor with the blocks of $x$. This corresponds to re-interpreting the set of indices $\mathcal{I}$
as the set of seeds of $\prng$. This replaces the need for $|\mathcal{I}|=2^b$ from construction $2$ by a much smaller $\mathcal{I}$ of size linear in $b$ and $\lambda$.

Concretely, let $\prng: \mathcal{Q} \to \{0,1\}^b$ be a pseudo-random number generator, with an arbitrarily large keyspace $\mathcal{Q}$.
Let an evaluation of $\prng$ on key $i$ as $\prng(i, b)$ denote that the $\prng$ stretches the input to length $b$.
For practical purposes we can think of $\mathcal{Q} = \{0,1\}^{128}$ if, for example, we instantiate $\prng$ with AES (in counter mode) with $128$-bit keys.
Let $\mathcal{I}$ be a subset of $\mathcal{Q}$. 
For construction $3$ we define $d_3(i,x)$ to be $\parity(\andreduce(x_i\xor \prng(i,b)))$.

Because we do not use a uniformly distributed mask on each block we do not get $1/2$-collision boundedness. However, we can get arbitrarily close to that by increasing the size of $\mathcal{I}$.
In particular $|\mathcal{I}|$ need only grow linearly in $b$, as shown in the following theorem.
\begin{theorem}\label{thm:construction3}
If $\prng$ is a random oracle, then given any $q>1/2$, there exists a choice of $|\mathcal{I}|$ such that with probability $1-2^{\sigma}$ (over the randomness of $\prng$), construction $3$ preserves $q$-collision boundedness. Specifically, it suffices to take
\begin{equation*}
|\mathcal{I}|\geq \frac{1}{2(q-1/2)^2}(\sigma + b + 1).
\end{equation*}
\end{theorem}

\begin{proof}
Unlike in the proof of Theorem~\ref{thm:construction2} we will make use here of the fact that the hypothesis in Theorem~\ref{thm:collision} is only required to hold for sufficiently large $\lambda$.

Let $m=\prng(i,b)$. Here we must show that with all but negligible probability, for sufficiently large $\lambda$, $x_j \neq x'_j$ implies $e_j:= d_2(x_j) \xor d_2(x_j') = \parity(\andreduce(x_j \xor m)) \xor \parity(\andreduce(x'_j \xor m))$ is equal to one with probability at least $1-q$.

We will show that with the choice of $|\mathcal{I}|$ given in the statement, the above will hold with probability $2^{-\sigma}$.

Given $x_j\neq x'_j$, let $y=x_j$ and $y'=x'_j$ to avoid extra subscripts. Let $y_l$ and $m_l$ be the $l$th bit, with \emph{one} indexing, of $y$ and $m$ respectively. Leaving AND implict (like multiplication) and using $\sum$ to denote XOR, we can rearrange the definition of $e_j$ as follows.
\begin{equation*}
\begin{split}
e_j = \sum_{l=1}^{b/2}\big(&m_{2l-1}\left(y_{2l}\xor y'_{2l}\right)\xor \\[-.1ex]
&m_{2l}\left(y_{2l-1}\xor y'_{2l-1}\right)\xor 
y_{2l}y'_{2l-1}\xor y_{2l-1}y'_{2l}\big)
\end{split}
\end{equation*}
Let $v(y,y')$ be the vector with entries $y_{l}\xor y'_{l}$ for all $l\in (1,...,b)$ plus one entry containing $\sum_{l=1}^{b/2} y_{2l}y_{2l-1}\xor y_{2l-1}y_{2l}$. Note that there are only $2^{b+1}$ possible values for $v$.

Now $e_j$ is a function of $v$ and $m$, we write it as $e(m,v)$. For a fixed value of $v$ let $q(v)$ be the fraction of the key space for which $e(m,v)=1$.
Note that this value has distribution $\text{Bin}(|\mathcal{I}|,1/2)$ with respect to the randomness of $\prng$. Therefore by a Chernoff bound we have that
$\PP(q(v) < p) \leq e^{-2(q-1/2)^2|\mathcal{R}|}$.
As there are $2^{b+1}$ possible values of $v$, a union bound over $v$ yields
$\PP(\exists v \textrm{ s.t. } q(v)<p)\leq 2^{b+1}e^{-2(q-1/2)^2|\mathcal{R}|}$.
Rearranging, it follows that it suffices to take
  $|\mathcal{I}|\geq \frac{1}{2(q-1/2)^2}(\sigma + b + 1)$.
\end{proof}

One should think of the $\sigma$ as a statistical security parameter, thus $40$ would be a standard choice. 

The prng can be thought of as a fixed function on the set $\mathcal{I}$. We need this function to have the property that any lie $v$ will be caught with probability at least $p$. This might not need to be a random oracle, but we can not prove any fixed function works, thus we instead show that a randomly selected function works with high probability.

However, it is important here that the randomness for the $\prng$ and the parameter $b$ are not chosen adversarially. If they are then the result could still be recovered by increasing $\sigma$ by however many bits of information about $b$ and the output of $\prng$ the adversary was able to control. We will use $\sigma=40$ when presenting our results.

It would be convenient if given a specific $b$ and $\prng$ we could check whether the resulting construction preserves $q$-collision boundedness. Unfortunately, the problem of determining whether this is the case is as hard as the learning parity with noise problem~\cite{LPN}, which is conjectured to be hard.

This does not rule out the idea of replacing the $\prng$ with a process that generates an output that is specially structured to guarantee preservation of $q$-collision boundedness. Indeed this is done in the analogous construction 4 of Section~\ref{sec:arithmetic} over large fields. However we were unable to find such a construction in the binary case.

The expansion of $\prng$ requires $O(b)$ gates, the evaluation of $h$ requires $O(n/b)$ gates, and $f$ requires $n/2$ AND gates and $3n/2$ XOR gates as in construction $2$. Thus by taking $b\approx \sqrt{n}$ the total cost is $n/2+O(\sqrt{n})$ AND gates and $3n/2+O(\sqrt{n})$ XOR gates.

However, in practice, as $O(n/b)$ is small compared to $n/2$ once $b$ is moderately large we advise taking $b\approx \min(\sqrt{n},1024)$ so that for large $n$ the size of $|\mathcal{I}|=O(\sigma+b)$ does not become prohibitive.

The choice of $q$ is somewhat arbitrary but it is a trade-off between wanting something close to $1/2$ whilst not wanting $|\mathcal{I}|$ to be too large. Taking $q=5/8$ is the compromise we work with.

With $\sigma=40$, $q=5/8$ and $b=1024$ we have $|\mathcal{I}|=34080$ indices. We will explore these values more in Section~\ref{sec:experiments}.

\section{PVC Committed MPC From Indexed \\ Hashes}\label{sec:integritycheck}

In this section we introduce PVC commitments and the required properties for them to be secure, instantiate them using indexed hash functions, and propose a protocol for committed MPC with PVC security that directly leverages PVC commitments.
We will start by defining what a PVC commitment scheme is, then we will explain how to construct one using a collision bounded indexed hash function. We will express the guarantees provided in a theorem and assess how the computational cost of the scheme depends on the indexed hash function. Throughout $\sk,\pk$ is a public key pair belonging to the committing party that can be thought of as the identity of the input, it should only be used by one input. It is important this public key is associated to the committer (possibly by being signed with another key) by anyone to whom the verifier wishes to prove cheating, e.g. a regulatory authority. The values $i\in \mathcal{I}$ and $r \in \mathcal{R}$ will be randomly chosen as inputs to provide security. For simplicity, we omit the security parameter $\lambda$ in some of our statements, and when we say that an adversary can not succeed at a task, we mean that they stand a negligible chance of doing so.

\subsection{Definitions}
We now define PVC commitments in terms of the three functionalities mentioned above.

\begin{definition}\label{def:secure-pvc-commitment}
A PVC commitment scheme with covert security parameter $p\in [0,1]$ consists of three functions $\pvccommit$, $\assert$, $\checkfun$, the last of which is deterministic, satisfying four security properties defined below (correctness, general binding property with parameter $p$, hiding property, and defamation freeness).
\end{definition}

Let us first describe the form of the three functions $\pvccommit$, $\assert$, and $\checkfun$.
A commitment function which commits to a value $x$,
\begin{equation*}
c=\pvccommit(x,\sk,r).
\end{equation*}

An assertion function which is applied to the $x$ we later wish to check was committed to,
\begin{equation*}
a=\assert(x,\sk,r;i,\pk).
\end{equation*}

And a checking function, which interprets the output from the other two functions,
\begin{equation*}
\out = \checkfun(c,a,\pk).
\end{equation*}
With output satisfying $\out \in \{\valid, \cheated, \inconclusive\}$.

Intuitively, $\out = \valid$ means that the commitment opened to the expected value,
$\out = \cheated$ means that the check did not pass because the committed and asserted values do not match,
and $\out = \inconclusive$ denotes situations where the result of the verification is inconclusive because
of a malformed message, or more generally an abort by the committer. This latter situation
can not be avoided in general when evaluating PVC commitments in MPC, as a corrupted committer could send invalid
messages or stop responding, similar to the role of aborts in MPC security with aborts. The first of the properties is correctness.
\begin{definition}[Property 1: Correctness]
For any $i,r,x$ and valid key pair $\sk,\pk$, if $c=\pvccommit(x,\sk,r)$ and $a=\assert(x,\sk,r;i,\pk)$ then
$\checkfun(c,a,\pk)=\valid$.
\end{definition}
The second is binding, a guarantee that a cheating committer will be caught with reasonable probability. P1 can avoid being caught cheating by refusing to sign anything, this is fine so long as they can not possibly get a $\valid$ result either. Thus we require that they be caught with probability $p$ only conditioned on the result not being inconclusive. A simple version of this is the following.

\begin{definition}[Honest Binding]
No polynomial time adversary can find $x,\sk,r,x',\sk',r',\pk$ such that (i) $x\neq x'$ and (ii) if $i \leftarrow \mathcal{I}$, $c=\pvccommit(x,\sk,r)$, $a=\assert(x',\sk',r';i,\pk)$ and
$\out=\checkfun(c,a,\pk)$
then $\PP(\out=\inconclusive)< 1$ and
\begin{equation*}
\PP(\out=\cheated|\out\neq \inconclusive)<p
\end{equation*}
\end{definition}

The above allows us to prove PVC security with parameter $p$ only if the commitment is made honestly. If the commitment might be arbitrarily generated then we need the following strictly stronger version of binding. As this version is stronger it is the only one we include in the definition of a PVC commitment scheme, the previous definition will be referenced later in proofs though.

\begin{definition}[Property 2: General Binding]\label{def:generalbinding}
No polynomial time adversary can find $x,\sk,r,x',\sk',r',\pk$ and $c$ such that (i) $x\neq x'$ and (ii) if $i \leftarrow \mathcal{I}$, $a=\assert(x,\sk,r;i,\pk)$, $a'=\assert(x',\sk',r';i,\pk)$,
$\out=\checkfun(c,a,\pk)$, and $\out'=\checkfun(c,a',\pk)$
then \newline $\PP(\out=\inconclusive)< 1$, $\PP(\out'=\inconclusive)< 1$ and
\startcompact{small}
\begin{align}
&\PP(\out=\cheated|\out\neq\inconclusive) \nonumber \\
+&\PP(\out'=\cheated|\out'\neq\inconclusive)<p \label{eqn:generalbinding}
\end{align}
\stopcompact{small}
\end{definition}
To see this is stronger, note that if a scheme is not honestly binding the same counterexample but with $c=\pvccommit(x,\sk,r)$ will show it is not generally binding.

The final two properties prevent the verifier from cheating, so consider $\sk,\pk$ to be fixed. It is useful to define an oracle $\mathcal{O}_\sk(x)$ which when called samples $r \leftarrow \mathcal{R}$ and returns
\begin{equation*}
\pvccommit(x,\sk,r) \nonumber
\end{equation*}
and
\begin{equation*}
\left(\assert(x,\sk,r;i,\pk)\right)_{i\in \mathcal{I}}.
\end{equation*}

The third property is the hiding property which guarantees the verifier can not learn anything about $x$ from the outputs of $\pvccommit$ or $\assert$.
\begin{definition} [Property 3: Hiding]\label{def:schemehiding}
For any $x,x'$ and polynomial time adversary $\adv$
\begin{equation*}
\PP(\adv(\mathcal{O}_\sk(x))=1)=\PP(\adv(\mathcal{O}_\sk(x'))=1)+\text{negl}(\lambda).
\end{equation*}
\end{definition}

The final property is defamation freeness which guarantees the verifier can not frame an honest committer.
\begin{definition}[Property 4: Defamation Freeness]
No polynomial time adversary can choose an $x$ and then when given $\mathcal{O}_\sk(x)$ find $c$ and $a$ such that
\begin{equation*}
\checkfun(c,a,\pk)=\cheated
\end{equation*}
\end{definition}

Note it is important that each secret key is only used for one choice of $x,r$. This could be achieved by deriving the secret key from $(x,r)$ by a one way function (possibly with extra randomness).

\subsection{PVC commitment from indexed hashes}
\label{sec:pvcindexed}
Let $H$ be an indexed hash function with index space $\mathcal{I}$ and randomness space $\mathcal{R}$. Let $m_{\sgn(\sk)}$ denote $m$ together with a signature of $m$ by secret key $\sk$. Consider the following three functions.

\begin{equation*}
\pvccommit(x,\sk,r)= ((H(i,r,x))_{i\in \mathcal{I}})_{\sgn(\sk)}
\end{equation*}

\begin{equation*}
\assert(x,\sk,r;i,\pk)= 
\begin{cases}
(i, H(i,r,x))_{\sgn(\sk)} & \text{if } (\sk,\pk) \text{ is a valid keypair} \\
\bot & \text{otherwise}
\end{cases}
\end{equation*}
For $\checkfun$ let $G$ be the event that the signatures are valid.
\begin{equation*}
\checkfun(c,a,\pk)=
\begin{cases}
\valid & \text{if $G$ and } c[a[0]]=a[1] \\
\cheated & \text{if $G$ and } c[a[0]]\neq a[1] \\
\inconclusive & \text{Otherwise}
\end{cases}
\end{equation*}

We require one slightly unusual property of the signature scheme. This is a technicality, as (a) lots of schemes have this property and (b) in the next subsection we will introduce a computational optimization which has the side effect of guaranteeing this property from any scheme.
\begin{definition}
Call a signature scheme \emph{discrimination resistant} if no polynomial time adversary can find $m,m',\sk$ and $\pk$, ($(\sk,\pk)$ not necessarily a valid key pair), such that $m_{\sgn(\sk)}$ and $m'_{\sgn(\sk)}$ are valid with non-negligibly different probabilities.
\end{definition}
We also require that the signature scheme has a deterministic verification function. This could be lifted at the expense of complicating the definitions with extra negligible terms. However, whilst not implied by the definition of a signature scheme, all the most popular schemes satisfy this assumption so we will make it for simplicity.

\begin{restatable}[]{theorem}{secpvccom}\label{theorem:secure-pvc-commitment}
If $H$ is hiding and $q$-collision bounded and the signature scheme has deterministic verification and is discrimination resistant, then the above functions form a PVC commitment scheme with covert security parameter $p=1-q$ (Definition~\ref{def:secure-pvc-commitment}).
\end{restatable}

The proof of this theorem is given in Appendix~\ref{app:proofcommitmentscheme}.
\subsection{PVC Committed MPC from a PVC \\ commitment scheme}
\label{sec:pvcmpc}
In this section we define formally PVC committed MPC, for the two party case,
and propose protocols to efficiently realize this functionality,
which corresponds to the intuitive idea from Figure~\ref{fig:basic_diagram}.

We follow the definitions by \citet{DBLP:conf/asiacrypt/AsharovO12} to prove PVC security of our protocols.
This involves proving (i) simulatability (in the ideal vs. real worlds framework) for the covert security part,
along with (ii) accountability and (iii) defamation freeness for the public verifiability.
For (ii) and (iii) we use the definitions by Asharov and Orlandi and for (i) our ideal world is
presented in detail in Appendix~\ref{sec:ideal} as an extension of theirs, to handle the commitment phase.
Without loss of generality, we describe our ideal world for only two parties $\A$ and $\B$.
Moreover, as in our protocols, the first party gets malicious security, while the second party gets PVC security. This matches the guarantee in the generic PVC protocol by Hong et al~\cite{pvc} that we use in the experimental evaluation.

Our ideal world is parameterized by two
values $p_{\texttt{exec}}, p_{\texttt{commit}}\in[0,1]$
denoting lower bounds on the probabilities with which $\A$
can get caught when (i) cheating in the protocol execution
and (ii) breaking the commitment, respectively.
Note that Asharov and Orlandi only formalize (i), 
and they denote $p_{\texttt{exec}}$ as $\epsilon$.
Moreover, our ideal world is parametrized by an arbitrary distribution
$\mathcal{E}$ with we refer to as \emph{the environment} (this is similar to the notion
used in the UC framework).
A sample from the environment is included in the parties' view as an
auxiliary input that is received only {\em after} the commitment phase has finished.
This limits the ability of the ideal world adversary (the simulator) to
rewind the adversary beyond the commitment phase (similar to the role of the environment in UC),
and models information that the adversary might get after committing.

We summarize the ideal world execution next. First, $\A$ receives its
prescribed input and commits to it (if honest)
or an arbitrary value (if corrupted) by sending it to 
the trusted party. This constitutes
the commitment phase, and captures the situation
where $\A$ commits to using an input, e.g., an ML model,
to be used at an undetermined time in the future
in a secure computation with a second party $\B$.
Then, party $\A$ receives an input from the environment,
in the form of a sample from $\mathcal{E}$, which is also
given to $\A$ in the real world, as explained above.
This determines the beginning of the secure computation phase, which starts with $\B$
receiving its prescribed input and with $\A$ 
notifying the trusted party of their desire to cheat in the execution. This attempt will succeed with probability
$1-p_{\texttt{exec}}$, in which case $\A$
gets to completely break the protocol, i.e. learn $\B$'s input and choose their output. If $\A$ fails, $\B$ receives output $\texttt{corrupted}$. If a corrupted $\A$
decided to not cheat in this way, they still get a 
chance to cheat in switching the input of the 
secure computation from the committed value $w$ to a different one.
If this attempt fails (which happens with probability at least $p_{\texttt{commit}}$),
$\B$ gets notified.
For simplicity in the presentation we allow $\A$ to abort after receiving their output, and before $\B$ gets to
observe theirs, but this assumption can be lifted by ensuring that in the underlying PVC protocol $\B$ gets the output first.

\begin{figure}[h]
\small
\begin{framed}
\small
\raggedright
\hspace{-0.1cm}\noindent {\bf Public Parameters:} A PVC commitment scheme (Definition~\ref{def:secure-pvc-commitment}) and a public key $\pk$.

 \hspace{-0.1cm}\noindent {\bf Inputs:} input $x$ and secret key $\sk$ matching $\pk$.

 \hspace{-0.1cm}\noindent {\bf Outputs:} Commitment $c$.
 
\vspace{0.2cm}
 \hspace{-0.1cm}\noindent {\bf Algorithm:}
 \begin{enumerate}[leftmargin=3ex]
   \item Sample $r\gets \mathcal{R}$.
   \item Compute $c=\pvccommit(x,\sk,r)$.
   \item Store $r$ as a secret and return $c$.
\end{enumerate}
\end{framed}
\caption{PVC Committed 2PC (commitment algorithm).}
\label{fig:pvc-mpc-commit}
\end{figure}

\begin{figure}[h]
\begin{framed}
\small
\raggedright
 \hspace{-0.1cm}\noindent {\bf Parties:} $\A$, $\B$.

 \hspace{-0.1cm}\noindent {\bf Public Parameters:} A PVC commitment scheme (Definition~\ref{def:secure-pvc-commitment}), a commitment $c$, and a public key $\pk$.\\[1ex]
 The protocol uses a PVC secure protocol $\Pi$ offering PVC security to $\B$ and malicious security to $\A$.
 
 \hspace{-0.1cm}\noindent {\bf Inputs:} $\A$: x, r; $\B$: y.

 \hspace{-0.1cm}\noindent {\bf Outputs:} $\A: g_1(x, y)$; $\B: g_2(x, y)$, or a proof of cheating $a$.
 
\vspace{0.2cm}
 \hspace{-0.1cm}\noindent {\bf Protocol:}
 \begin{enumerate}[leftmargin=3ex]
 \item $\B$ samples $i\gets \mathcal{I}$.
 \item $\A, \B$ run $\Pi$ to compute $(o_1;o_2,a)=(g_1(x,y);g_2(x,y),\assert(x,\sk,r;i,\pk))$.
 \item $\B$ computes $\out=\checkfun(c, a,\pk)$ and\\\noindent
   If $\out = \valid$ $\longrightarrow$ accepts $o_2$ as $g_2(x, y)$.\\\noindent
   If $\out = \cheated$ $\longrightarrow$ accepts $a$ as proof of cheating.\\\noindent
   Otherwise $\longrightarrow$ aborts and sets result to $\emph{inconclusive}$.
   
\end{enumerate}
\end{framed}
\caption{PVC Committed 2PC for functionality $g(x, y) = (g_1(x, y), g_2(x, y))$ (integrity check).}
\label{fig:pvc-mpc-generic}
\end{figure}

\begin{figure}[h]
\begin{framed}
\small
\raggedright

Parties, inputs, outputs, and public parameters are as in Figure~\ref{fig:pvc-mpc-generic},
and the PVC commitment scheme is instantiated by an indexed hash function $H$ (as in Theorem~\ref{theorem:secure-pvc-commitment}).\\
\vspace{0.2cm}
 \hspace{-0.1cm}\noindent {\bf Protocol:}
 \begin{enumerate}[leftmargin=3ex]
 \item $\B$ samples $i\gets \mathcal{I}$ and $\tilde{r}\gets \mathcal{R}$.
 \item $\A, \B$ run $\Pi$ to compute $(g_1(x,y),h(m|\tilde{r});g_2(x,y),m)$, where $m = (i,H(i,r,x))$.
 \item $\A$ computes $s=\sign(h(m|\tilde{r}),\sk)$ and sends it to $\B$.
 \item $\B$ aborts if $s$ is not the valid signature of $h(m|\tilde{r})$.
 \item $\B$ computes $\out=\checkfun(c, a,\pk)$ and\\\noindent
   If $\out = \valid$ $\longrightarrow$ accepts $o_2$ as $g_2(x, y)$.\\\noindent
   If $\out = \cheated$ $\longrightarrow$ accepts $a$ as proof of cheating.\\\noindent
   Otherwise $\longrightarrow$ aborts and sets result to $\emph{inconclusive}$.
   
\end{enumerate}
\end{framed}
\caption{PVC Committed 2PC for functionality $g(x, y) = (g_1(x, y), g_2(x, y))$ (Optimized integrity check).}
\label{fig:pvc-mpc-optimized}
\end{figure}

Let $(\pvccommit,\assert,\checkfun)$ be a PVC commitment scheme with parameter $p$. Let $\texttt{Blame}_{\texttt{commit}}$ be the function which when given a view of $\B$ (honestly) running the protocol in Fig.~\ref{fig:pvc-mpc-generic}, in which $\out=\cheated$ returns the commitment $c$ and the resulting $a$ and otherwise returns $\bot$. Let $\texttt{Judgement}_{\texttt{commit}}$ be the function $\checkfun$ with the public key of $\A$ hard coded. Let $\texttt{Commit}$ be the commitment algorithm in Fig.~\ref{fig:pvc-mpc-commit} and $\mathcal{P}$ be the protocol in Fig.~\ref{fig:pvc-mpc-generic}, with $\Pi$ instantiated with the protocol of Hong et al~\cite{pvc}. Finally, let
$\texttt{Blame}_{\texttt{exec}}$ and $\texttt{Judgement}_{\texttt{exec}}$ be the blame and judgement functions from $\Pi$,
and define $\texttt{Blame}(x)$ to be $\cheated$ if either $\texttt{Blame}_{\texttt{exec}}(x)$
or $\texttt{Blame}_{\texttt{commits}}(x)$ equals $\cheated$, and analogously for
a function $\texttt{Judgement}$. We are now ready to state our main result.

\begin{restatable}[]{theorem}{mainthm}\label{thm:main}
The quadruple $\left(\texttt{Commit},\mathcal{P},\texttt{Blame},\texttt{Judgement}\right)$ securely computes $g$ with committed first input in the presence of a malicious $\A$ or a covert $\B$ with $p/2$-deterrent and public verifiability.
\\[1ex]
If $\pvccommit$ and $\checkfun$ are used as given in the previous section then we can replace the $\mathcal{P}$  with the protocol in Fig.~\ref{fig:pvc-mpc-optimized} and still have the same security guarantee.
\\[1ex]
Furthermore, if in either case it can be guaranteed that $\A$ is honest in running the commitment algorithm in Fig.~\ref{fig:pvc-mpc-commit}, then the deterrent factor improves from $p/2$ to $p$.
\end{restatable}

\myparagraph{Non-committed output at no risk.}
P1 can in the above ideal world, and thus in the protocol, get $g_1(x',y)$ for a non-committed $x'$ at no risk by aborting afterwards. This could be avoided by opening up the PVC blackbox and holding back this output until P2 has checked the result of assert (or optimized equivalent).

\myparagraph{Computational costs.} The cost of the commit operation in the clear is computing $H$, $|\mathcal{I}|$ times. The cost of the $\assert$ is dominated asymptotically by the cost of computing $H$ once i.e. requires $n/2+o(n)$ AND gates. The $\checkfun$ are $O(1)$ and relatively very cheap.

\section{Evaluation}\label{sec:experiments}

\label{sec:experiment_results}

\begin{table*}[]
\resizebox{0.75\textwidth}{!}{
\small
\begin{tabular}{@{}ccc|ccc@{}}
\toprule
\multicolumn{1}{c}{}            & \multicolumn{1}{c}{SHA3-256}      & \multicolumn{1}{c|}{LowMCHash-256}      & \multicolumn{3}{c}{Ours}                                                                                                                                                                                     \\ \midrule
\multicolumn{1}{c}{No. of bits} & \multicolumn{1}{c}{\# of ANDs} & \multicolumn{1}{c|}{\# of ANDs} & \multicolumn{1}{c}{\# of ANDs} & \multicolumn{1}{c}{\begin{tabular}[c]{@{}c@{}}Improvement\\ over SHA3-256\end{tabular}} & \multicolumn{1}{c}{\begin{tabular}[c]{@{}c@{}}Improvement over\\ LowMCHash-256\end{tabular}} \\ \midrule
$2^{14}$                        & $6.14\times10^5$                        & $2.32\times10^5$                         & $5.17\times10^4$                         & $12\times$                                                                           & $4\times$                                                                            \\
$2^{18}$                        & $9.29\times10^6	$                      & $3.65\times10^6$                        & $1.90\times10^5$                       & $49\times$                                                                           & $19\times$                                                                           \\
$2^{22}$                        & $1.48\times10^8$                    & $5.84\times10^7$                       & $2.29\times10^6$                        & $65\times$                                                                           & $25\times$                                                                           \\
$2^{26}$                        & $2.37\times10^9$                   & $9.34\times10^8$                      & $3.42\times10^7$                       & $69\times$                                                                           & $27\times$                                                                           \\
$2^{30}$                        & $3.79\times10^{10}$                   & $1.49\times10^{10}$                    & $5.39\times10^8$                      & $70\times$                                                                           & $28\times$                                                                           \\ \bottomrule
\end{tabular}}
\caption{Analytical comparison of the number of AND gates (circuit size $|\mathcal{C}|$) for the $\assert$ functionality using LowMCHash-256 and SHA3-256 with our scheme. These values are for a single call to $\assert$ i.e. using a single index for our scheme. Here $p_c = 1/2$}
\label{tab:and}
\end{table*}

Here we compare our method for committed MPC to the baseline using SHA3-256. We evaluate both computation time and communication for the $\assert$ functionality  as the size of the input $n$ increases. We also analytically compare our method against the hash function based on LowMCHash-256, an MPC friendly hash \cite{albrecht2015ciphers}. Finally, we evaluate the practicality of our proposed scheme in terms of the compute requirement for the committer performing the commitment using the $\pvccommit$ functionality and the size of the commitment. As a result, we show, for our scheme: (a) Verification ($\assert$ functionality) in MPC is significantly faster than optimized standards such as SHA3-256 as well as MPC optimized hashes such as LowMCHash-256; (b) The size of the commitment is practical; (c) The computation required from the committer ($\pvccommit$ functionality) is practical. We use the circuit sizes and real experimental data for (a). Similarly, we analyze the size of the commitment to prove (b) and use actual computation time data to show (c). We begin by describing the experimental and implementation details.

\myparagraph{Experimental Settings.}
The experiments were executed on two Azure D32s v3 machines running Ubuntu 16.04, equipped with Intel Xeon E5-2673 v4 2.3GHz processors and 128 GB RAM. The machines were hosted in the same region with a bandwidth of 1.7 GB/s and an avg. latency of 0.9ms, representative of a LAN setting.

\myparagraph{Implementation.}
We use the EMP-toolkit \cite{emp-toolkit} to implement our secure protocols as well as the baselines. In particular, we use the PVC framework of Hong et al. \cite{hong2019covert}, which makes use of garbled circuits. We set the covert security parameter $p_c$ of this underlying implementation to $1/2$. Note this is different from the covert security parameter $p$ used in our scheme. Since $p\leq 1/2$, $p_c$ could be set to  $1/2$. As one could infer, the effective covert security parameter for our scheme with this implementation would be $min(p,p_c)$.

\myparagraph{Baselines.}
We use two baselines for comparison: SHA3-256 and LowMCHash-256.
For SHA3-256, we use the sponge framework \cite{bertoni2008indifferentiability} with an  input block size of $1600$. Using the standard security parameters we get the rate as 1088 and the capacity as $512$. This results in a computation cost of $\sim35$ AND gates per input bit. For LowMCHash-256, we use LowMC permutations together with the sponge framework using an input block size of 512. We reserve 256 bits for the rate and another 256 bits for the capacity (128 bit security). This results in $\sim14$ AND gates per input bit. LowMC is relatively new and has been shown to be susceptible to attacks \cite{dinur2015optimized}. However, we include it in this comparison, it being one of the most MPC optimized hashing schemes for Boolean circuits. Both these baselines are implemented on the top of the EMP-toolkit's PVC framework.

\myparagraph{Our scheme.} For our scheme we implement the idea around PVC commitment from indexed hashes as described previously. We use Construction $3$ in section~\ref{sec:defs}. In particular we implement the $\assert$ functionality in MPC (and $\pvccommit$, $\checkfun$ in the clear). Our scheme costs $\sim0.5$ AND gate per input bit. In order to effect signed public verifiability during the $\assert$ phase, we use SHA3-256 to commit the hash corresponding to the input and the index. We summarise the parameters for our scheme in Table~\ref{tab:params}. Note that the table reports the covert security parameter of the commitment scheme for the honest committer case. In the general case this parameter's value would be $p/2 = 3/16$. Similar to the baselines, our scheme is also implemented on top of the EMP-toolkit's PVC framework.

\begin{table}[]
\small
\resizebox{0.95\columnwidth}{!}{
\begin{tabular}{@{}ccc@{}}
\toprule
Description                         & Symbol   & Value                 \\ \midrule
\multicolumn{3}{c}{Our Indexed Hash Function}                          \\\midrule
Length of the input $(|x|)$           & $n$ & no. of bits (variable)\\
Pseudorandom number generator       & $prng$   & AES (counter mode)    \\
Underlying collision resistant hash & $h$      & SHA3-256              \\
Statistical security parameter      & $\sigma$ & 40                    \\
Collision boundedness parameter     & $q$      & 5/8                   \\
Block size                          & $b$      & $min(\sqrt{n}, 1024)$ \\\midrule
\multicolumn{3}{c}{Our PVC Commitment}                                     \\\midrule
Covert security parameter           & $p$      & $1-q = 3/8$             \\\midrule
\multicolumn{3}{c}{Underlying PVC 2PC Protocol (EMP-PVC) For $\assert$}         \\\midrule
Covert security parameter           & $p_c$    & $1/2$                 \\ \bottomrule
\end{tabular}}
\caption{Parameters used in our experiments}
\label{tab:params}
\end{table}

\subsection{Analytical Performance}

Table~\ref{tab:and} compares the circuit size $|\mathcal{C}|$ (no. of AND gates) for the $\assert$ functionality for the LowMCHash-256 and SHA3-256 baselines with our scheme. As we increase the size of the input, the scheme starts to show its full potential. For a small input size, the initial overhead of signing the commitments and the index tends to shadow the improvement. But as we increase the size of the input, we can see a marked $70\times$ improvement over SHA3-256 and $28\times$ improvement over LowMCHash-256. We show that these improvements directly translate into real world experiments, when compared against the actual implementation of SHA3-256, in section~\ref{sec:exp}.

\subsection{Experimental Performance}
\label{sec:exp}

\myparagraph{Running time for $\assert$.} Table~\ref{tab:time} shows the running time for executing the $\assert$ functionality to verify the commitments using SHA3-256 and our scheme. As we increase the size of the input to practical sizes, we observe that our scheme is $60\times$ faster than the SHA3-256 baseline. This is directly correlated with the $70\times$ improvement in the circuit sizes above. We do not perform actual experiments with LowMCHash-256, but it is similarly expected to be around $25\times$ slower than our scheme as indicated by the circuit sizes. Also, in practice, nothing prohibits us from replacing our underlying collision resistant hash $h$ with LowMCHash to amplify this improvement. We compare the 

\myparagraph{Communication for $\assert$.} Table~\ref{tab:comm} (Appendix \ref{app:commassert}) shows the amount of communication needed for executing the $\assert$ functionality using SHA3-256 and our scheme. We observe that our scheme requires $36\times$ less communication for the committer and the verifier than the SHA3-256 based baseline.

\begin{table}
\centering
\small
\resizebox{0.95\columnwidth}{!}{
\begin{tabular}{@{}cccc@{}}
\toprule
\multicolumn{1}{c}{No. of bits} & \multicolumn{1}{l}{Ours (s)} & \multicolumn{1}{l}{SHA3-256 (s)} & \multicolumn{1}{l}{Improvement}   \\ \midrule

$2^{14}$                        & 0.07                         & 0.57                          & $8\times$                       \\

$2^{18}$                        & 0.22                         & 8.16                          & $36\times$                      \\

$2^{22}$                        & 2.67                         & 133.23                        & $50\times$                      \\

$2^{26}$                        & 39.14                        & 2200*                       & $56\times$                      \\

$2^{30}$                        & 590.70                       & 35500*                      & $60\times$                      \\  \bottomrule
\end{tabular}}
\caption{Comparison of running time for SHA3-256 baseline and our scheme executing the $\assert$ functionality. Here $p_c = 1/2$. * means estimated via extrapolation}
\label{tab:time}
\end{table}

\myparagraph{Computation load for $\pvccommit$.} 
In Table~\ref{tab:comp}, we show the number of indices $|\mathcal{I}|$ for the commitment that needs to be computed alongside the size of the entire commitment that a committer needs to prepare in order to commit its input. In these experiments $p=3/8$ ($q=5/8$) and block size $b = min(\sqrt{n}, 1024)$. The size of the commitment results in a very limited communication and space requirement. Block size limit of 1024 bits, limits the size of the commitment to just 1.09 MB. We use the formulation, upon ceiling to the next nearest integer, defined in Theorem~\ref{thm:construction3} to compute $|\mathcal{I}|$. In Figure~\ref{fig:indices} (Appendix~\ref{app:numind}), we plot this formulation for $\sigma = 40$, $b = 1024$ and different values of q (and the covert security parameter $p$ i.e $1-q$) to show how the number times $|\mathcal{I}|$ that the committer needs to compute $H$ varies with the security parameters.

In Table~\ref{tab:comp} we also show the computation load of the committer for committing its input. In particular, we evaluate the time needed to perform the $\pvccommit$ functionality. This only needs to be performed once for a given input, in the clear. We see that the computation load is very limited even for large input sizes. For these results we use only a single process, however this computation is trivially parallelizable. Several hashes can be computed in parallel. For example a 128 threaded implementation should enable the committer to commit $2^{30}$ bits in less than 3 minutes. Furthermore, we perform these computations in Python using standard libraries and there is scope for further significant optimization by using a low-level language.

\begin{table}[]
\small
\resizebox{0.95\columnwidth}{!}{
\begin{tabular}{@{}cccc@{}}
\toprule
\multicolumn{1}{c}{No. of bits} & \multicolumn{1}{c}{\# of Hashes}  & \multicolumn{1}{c}{\begin{tabular}[c]{@{}c@{}}Size of the\\ Commitment (MB)\end{tabular}}  & Time   \\ \midrule
$2^{14}$                        & 5408                     & 0.17             & 1.61s  \\
$2^{18}$                        & 17696                      & 0.57         & 28.21s \\
$2^{22}$                        & 34080                       & 1.09         & 5.06m  \\
$2^{26}$                        & 34080                        & 1.09       & 36.92m  \\
$2^{30}$                        & 34080                       & 1.09       & 5.91h \\ \bottomrule
\end{tabular}}
\caption{Computation time and size of the commitment using our scheme for executing the $\pvccommit$ functionality. Here $p = 3/8$, statistical security parameter $\sigma = 40$ and block size $b=min(\sqrt{n},1024)$ where $n$ is the size of input.}
\label{tab:comp}
\end{table}

\section{Lower Bounds}\label{sec:lower}
In this section we provide lower bounds on how many AND gates are required for a collision resistant indexed hash function and an ordinary hash function. Recall that our construction $3$ from Section~\ref{sec:coverthashes} requires half an AND gate per bit of input.
In this section we show that
\begin{enumerate}
\item Construction $3$ is optimal amongst hash functions whose output's size is sublinear w.r.t. their input's size (Corollary~\ref{cor:lowercovert}).\label{item:lb1}
\item  For ordinary hash functions we  show that every collision resistant hash function requires at least one AND gate per input bit (Proposition~\ref{prop:lowerordinary}).\label{item:lb2}
\item Assuming that we want our hash functions to be hiding, we show, both for indexed and ordinary hash functions, that allowing their output to be large does not help much to reduce the number of required nonlinear gates (Proposition~\ref{prop:lowerhiding}).\label{item:lb3}
\end{enumerate}

Moreover, although we state the above results in terms of Boolean circuits, it is not hard to see that the arguments extend to any field. The following lemma and corollary correspond to item~\ref{item:lb1} above. The proof, given in Appendix~\ref{app:lb}, constructs an algorithm to find a $1$-collision on any $H$ with small set of nonlinear gates by casting that problem as that of solving a linear system $S$ on $\mathbb{F}_2$, and showing that $S$ always has a solution. Recall that indexed hash functions have three inputs $i, r, x$, in the statement by {\em main input} we mean $x$.

\begin{restatable}[]{proposition}{lowerpropone}\label{prop:lowercovert}
  Given any non-trivially collision bounded family of indexed hash functions $\{H_k\}_{k\in K}$ with $H_k$ given by the (polynomial size) circuit $C_k$ with $n$-bit main input, and $m$-bit output. With all but negligible probability over the generation of $k=G(\lambda)$, the circuit $C_k$ must have at least $\lceil (n-m)/2 \rceil$ nonlinear gates.
\end{restatable}

Note that in practice the lower bound on the nonlinear gate count will apply (with all but negligible probability) for any $\lambda$ large enough to be considered secure.
In particular we have the following corollary which says that, in order to beat our constructions asymptotically, an indexed hash function must have large output.

\begin{corollary}\label{cor:lowercovert}
  Any family of covertly collision resistant hash function circuits, indexed by $n$, with main input in $\{0,1\}^n$ must either have at least $n/2+o(n)$ nonlinear gates or must have output size that is not $o(n)$.
\end{corollary}

A stronger result can be achieved in the case of an ordinary secure hash function, by relying on the fact that they do not take auxiliary inputs. The idea of the proof is similar to that of Proposition~\ref{prop:lowercovert}, and is given in Appendix~\ref{app:lb}.

\begin{restatable}[]{proposition}{lowerproptwo}\label{prop:lowerordinary}
  Let $\{h_k\}_{k\in K}$ be a collision resistant family of hash functions with $h_k$ given by the circuit $C_k$ with $n$-bit input and $m$-bit output. With all but negligible probability with respect to the generation of $k=G(\lambda)$, the circuit $C_k$ must have at least $n-m$ nonlinear gates.
\end{restatable}

These results show that our constructions have asymptotically half the verification cost of the baseline with any ordinary secure hash function. However, recall that we designed our construction $3$ for the output of $H$ to be small, i.e. $o(n)$, for {\em efficiency and not security reasons}. One may thus wonder whether dropping this requirement allows to significantly overcome the above lower bounds. Next, we show that the answer is negative by leveraging the fact that we do require a hiding property {\em for security}, which we show implies a linear lower bound on the required number of AND gates.

The proof of the following result can be found in Appendix~\ref{app:lb}. It relies on the fact that if you have a small number of AND gates then only a small amount of the entropy in the randomness can affect their inputs. The rest of the randomness can not be used for hiding the output without giving too much leeway for finding collisions. Thus only a small amount of randomness and a small number of output wires from AND gates can hide the output. Thus the output must be effectively small and the above propositions can be applied.

\begin{restatable}[]{proposition}{lowerpropthree}\label{prop:lowerhiding}
  Suppose that $\{H_k\}_{k\in K}$ is a non-trivially collision bounded and hiding family of hash functions. Let $H_k$ be given by $C_k$ with an $n$-bit main input and $d$ nonlinear gates, then with all but negligible probability, $d\geq n/5$.
  Further if $|\mathcal{I}|=1$, then $d\geq n/3$.
\end{restatable}

\section{From Covert to Malicious Security}\label{sec:malicious}
A natural idea is to amplify the statistical guarantee of an indexed hash function $H$
by computing it at several indices. This would in turn lead to a PVC commitment scheme
with improved parameters
where $H$ is run on several indices.
More concretely, given a collision resistant indexed hash function $H$ we can compute an indexed hash function $H^{\kappa}$ with stronger security by computing $H$ $\kappa$ times with different indices. Formally, with $i_j\in \mathcal{I}$ for $j\in\{1,...,\kappa\}$
\begin{equation}
H^{\kappa}((i_j)_{j=1}^{\kappa},r,x)=\left(H(i_j,r,x)\right)_{i=1}^{\kappa}.
\end{equation}
This new function requires no more hashes to be prepared by the committer and, if $\{H_k\}_{k\in K}$ is $q$-collision bounded then $\{H_k^\kappa\}_{k\in K}$ is $q^\kappa$-collision bounded. However, it also requires $\kappa$ times as many AND gates (and XOR gates) to compute it. In this section, we present a construction that asymptotically requires no more AND gates than $H$ (and fewer XOR gates than $H^{\kappa}$) to achieve this higher security.

Let $E:\{0,1\}^w\rightarrow \{0,1\}^l$ be the encoding function of a $(\kappa-1)$ error detecting code.
All we require from $E$ is that if two messages $m,m'\in\{0,1\}^w$ then their codes, i.e. $E(m), E(m')\in \{0,1\}^l$ differ in at least $\kappa$ positions.
Split $x$ into $w$ words, $x_1,...,x_w$ each of length $\lceil n/w\rceil$, zero-padding $x$ as required. Let $x^1,...,x^{\lceil n/w\rceil}$ be the columns of the matrix whose rows are given by the $x_j$. Let $\tilde{x}_1,...,\tilde{x}_l$ be the rows of the matrix whose columns are given by $E(x^1),...,E(x^{\lceil n/w\rceil})$. Finally let
\begin{equation}
H^E\left((i_j)_{j=1}^{l},r,x\right)=\left(H(i_j,r,\tilde{x}_j)\right)_{j=1}^{l}.
\end{equation}

The following theorem follows from the structure of $H^E$ and the property of the error detecting code (proof in Appendix~\ref{app:malicious}).
\begin{restatable}[]{theorem}{maliciousone}
If $\{H_k\}_{k\in K}$ is $q$-collision bounded then $\{H_k^E\}_{k\in K}$ is $q^\kappa$-collision bounded.

Furthermore, the number of AND and XOR gates required to compute $H_E$ is $l\lceil n/w \rceil$ times the number of gates required per bit by $H$ plus $\lceil n/w \rceil$ times the number of gates required by $E$.
\end{restatable}

To make use of the above result we need an error detecting code $E$ that works on fairly large codewords and is easy to compute. We want it to be linear to keep the number of AND gates low, but we also do not want to introduce too many XOR gates. The following lemma provides such an encoding.

\begin{lemma}\label{lem:E}
Given $\rho,d \in \mathbb{Z}_+$, there exists a linear $2^d-1$ error detecting encoding $E:\{0,1\}^{\rho^d}\rightarrow \{0,1\}^{(\rho+1)^d}$ requiring $(\rho-1)((\rho+1)^d-\rho^d)$ XOR gates to compute.
\end{lemma}
\begin{proof}
Given a message $m\in \{0,1\}^{\rho^d}$, arrange the bits of $m$ in a $d$-dimensional cube. We index into $m$ with the notation $m[i_1,...,i_d]$. We extend $m$ by one in each dimension in turn by the following method. To extend $m$ by one in the dimension $j$, let $m[i_1,...,i_{j-1},\rho,i_{j+1},i_d]$ be the XOR of $m[i_1,...,i_{j-1},0,i_{j+1},i_d]$ through $m[i_1,...,i_{j-1},\rho-1,i_{j+1},i_d]$. The output of $E$ is just the contents of the resulting cube.

Let $m'$ be a different message, then for some choices of $i_j$ we have that $m[i_1,...,i_d]\neq m'[i_1,...,i_d]$. We can then deduce by induction that after $j$ dimensions have been extended there are at least $2^j$ points in the cuboids with final co-ordinates $i_{j+1},...,i_d$ on which $m$ and $m'$ differ. Thus once all directions have been extended the arrays $m$ and $m'$ differ in at least $2^d$ places and we have a $2^d-1$ error detecting code.

The $j$th extension requires $(\rho-1)(\rho+1)^{j-1}\rho^{d-j}$ XOR gates. Summing over all $j$ gives the result.
\end{proof}

Putting the above together we get a corollary which says there exists an asymptotically efficient protocol for maliciously secure commitment. Note that $\log 1/q$ is a statistical security parameter so can be thought of as a small constant, independent of $n$ and $\lambda$, in practice $\log_2 \log_2 1/q = 6$ should suffice.

\begin{corollary}\label{cor:maliciousworks}
Assume the existence of a collision resistant family of hash functions $\{h_k\}_{k\in K}$ with run time linear in input size and a random oracle $\prng$. Then there exists a $q$-collision bounded indexed hash function family with the following two properties. For a fixed security parameter, it can be computed with $n/2+o(n\log \log 1/q)$ AND gates and $(5/2+\lceil\log_2 \log_2 1/q\rceil)n+o(n\log \log 1/q)$ XOR gates. It requires $o(n\log_2 \log_2 1/q)$ information to be stored in order to be able to check any result.

Furthermore, if $\{h_k\}_{k\in K}$ is replaced by a family of random oracles then the resulting indexed hash function family is hiding.
\end{corollary}
\begin{proof}
Let $E$ be the encoding function given in Lemma~\ref{lem:E} with $d=1+\lceil\log_2 \log_2 1/q\rceil$ and $\rho=\lceil n^{1/3d} \rceil$. Let $H=\mathcal{C}_3(h_k)$ with $|\mathcal{I}|$ chosen to give collision resistance with parameter $1-\sqrt{1/2}$. Then $\{H_k^E\}_{k\in K}$ has all the required properties.
\end{proof}

We have not done any experiments with this idea, however from preliminary estimates of AND gate counts (with $q=2^{-2^6}$) we are confident that it offers no improvement for inputs of $10^6$ bits. If the choices of parameters were optimized we believe it would beat the baselines for $n=10^9$, though the cross over point depends on the baseline and choice of $h$ (and $\prng$).

This effectively recovers malicious security in the setting where the commitment is honestly generated, by the results of Section \ref{sec:integritycheck}. In fact, however, this method can recover malicious security in the presence of arbitrarily generated commitments too. As on all but at most one input (decided at commitment time) $H$ will catch cheating with probability $p/2$, it can be guaranteed that $H^E$ will catch cheating with all but probability $(1-p/2)^\kappa$. Thus for the not honestly committed case we need to only increase the choice of $d$ by one in the proof of Corollary \ref{cor:maliciousworks}.

\section{Arithmetic Circuits}\label{sec:arithmetic}
We have mainly focused on binary circuits because they are more flexible and there are more reasonably fast hash functions for them. However our main idea will also work to construct indexed hash functions to be computed in arithmetic circuits. As before our constructions are in terms of a secure hash function $h$ which could be implemented using MiMC~\cite{albrecht2016mimc} or any other arithmetic circuit hash function. We will assume this arithmetic is in a field $\mathbb{F}$.

Analogues of constructions $2$ and $3$ would work in this setting with XOR and AND gates replaced by ADD and MUL gates. Indeed, these would also work, with worse parameters, over arbitrary rings. These can be analysed analogously and relevant theorems deduced. However we will not detail these changes here and will instead provide a further development that was not possible in the binary case.

The idea of construction $4$ presented in this section is much like the analogue of construction $3$, however instead of using a $\prng$ to generate the random masks to be added to index,  we will generate them in a more structured fashion. Hence, construction $4$ still follows the blueprint given in Equation~\ref{eq:basicform}. The index space $\mathcal{I}$ will be a subset of $\mathbb{F}$, this requires the field to be moderately large and rules out this construction in the binary case.

As in Section~\ref{sec:coverthashes}, we have an even block size parameter $b$, and define the indexed hash by means of a digest function $d_4$ that takes an index $i$ and $b$ field elements as input and returns a single field element. Given $y$ a fixed block of $b$ elements denoted by $y_1,...,y_b$,
\begin{equation}
  d_4(i,y)=\sum_{j=1}^{b/2}(i^{2j-1}+y_{2j-1})(i^{2j}+y_{2j})
\end{equation}
The value of the hash is given by $\mathcal{C}_4(h,\lambda)(i,r,x) = h(r||i||\process_{b, d_4)}(i, x))$ given functions $\process$ and $h$, as described in Equation~\ref{eq:basicform} and Section~\ref{sec:coverthashes}.

The following Theorem states the guarantee of construction $4$. While its full proof is given in Appendix~\ref{app:arithmetic}, the basic idea is that there will be a collision so long as some vector determined from $x$ and $x'$ is not perpendicular to $(1,i,i^2,...,i^{b})$. The powers of $i$ come from the definition of $d_4$ and have been chosen (to replace the $\prng$) so that these vectors form Vandermonde matrices, thus any $b+1$ of them span and so at most $b$ are perpendicular to any given vector.

\begin{restatable}[]{theorem}{arithmeticthm}
  Construction $4$ is $b/|\mathcal{I}|$-collision bounded.
\end{restatable}

Note that this construction only works for fields larger than the block size $b$, but this is the case for a lot of standard hashes based in field arithmetic. If the field is very large then the covert security parameter can be made $\approx 1$ by taking $|\mathcal{I}|$ to be big. However this would be very impractical to prepare the hashes, and thus in that case it would be more practical to combine construction $4$ presented in this section with the amplification ideas from section~\ref{sec:malicious}.

\section{Conclusion}\label{sec:conclusion}
The standard simulation-based security definitions used in MPC allow a malicious adversary controlling one of the parties to provide arbitrary inputs. This leaves concerns related with input validity. In this paper, we introduced a method for securely committing an input in 2PC publicly verifiable covert (PVC) model for Boolean circuits. PVC security is valuable when the reputation of the committing party is at stake. Our methods are based on our introduction of indexed hashes and $q-$collision resistance and make use of the covert security guarantees and interactivity in MPC. Our work improves upon ordinary hash functions both in speed and communication. Our work is the first we are aware of to enable commitments in MPC for PVC security. We also extend our methods to the maliciously secure model and arithmetic circuits. 

Future work could evaluate these methods for certified prediction and for the maliciously secure variant with optimized parameters. There is also a gap to be closed between constructions and lower bounds if we allow large commitments. The requirements on $\prng$ are slightly inconvenient and a deterministic way to find vectors for construction 3, like in construction 4, would be useful. One could investigate if the commitment size could be reduced to $O(1)$ while maintaining half an AND gate per bit cost.

\begin{acks}
NA was supported by University of Oxford and Callsign. This work was done when AG was at The Alan Turing Institute (ATI) and Warwick University. AG and JB were supported by ATI under the EPSRC grant EP/N510129/1, and the UK Government’s Defence \& Security Programme. We also acknowledge ATI's support and generous provision of Azure cloud computing resources.
\end{acks}

\bibliographystyle{ACM-Reference-Format}

\ifthenelse{\equal{\nsubmission}{T}}
{\bibliography{Documents/references}}
{\bibliography{../Documents/references}}

\appendix

\section{Proofs from Section~4}\label{app:coverthashes}
\hidingtheorem*
\begin{proof}
Let $r\leftarrow \mathcal{R}$. Suppose that a polynomial time algorithm $A$ is given input $k$, $(H_k(i,r,x))_{i\in \mathcal{I}}$. For fixed $k$, the $H_k(i,r,x)$ are independent uniform random variables irrespective of the value of $x$ or $r$, so without querying the oracle the adversary can learn nothing about $x$ or $r$.

When the adversary requests the value of the random oracle on an input beginning with $r'\in S$ suppose it is also told whether or not $r'=r$.

When the adversary queries with $r'\neq r$ it learns nothing about $r$ except that $r\neq r'$. Thus the probability of using the right salt on the $j$th query is at most $1/(|\mathcal{R}|-j+1)$ and so the probability of querying the correct $r$ with $a$ guesses is at most $a/|\mathcal{R}|$. As the adversary has time for only polynomially many queries and $|\mathcal{R}|$ grows exponentially in $\lambda$ it will query with $r$ as the randomness with negligible probability.

Conditioned on $A$ never querying the correct randomness, its view is independent of $x$ and thus so is the probability of it outputting $1$.
\end{proof}

\section{Communication for $\assert$.}\label{app:commassert}
Table~\ref{tab:comm} shows the amount of communication needed for executing the $\assert$ functionality using SHA3-256 and our scheme. We observe that our scheme requires $36\times$ less communication for the committer and the verifier than the SHA3-256 based baseline.

\begin{table}
\centering
\begin{tabular}{@{}cccc@{}}
\toprule
\multicolumn{1}{l}{No. of bits} & \multicolumn{1}{l}{Ours (MB)} & \multicolumn{1}{l}{SHA3-256 (MB)} & \multicolumn{1}{l}{Improvement} \\ \midrule

$2^{14}$                        & 2.51                          & 19.93                          & $8\times$                       \\

$2^{18}$                        & 10.90                         & 300.34                         & $28\times$                      \\

$2^{22}$                        & 141.25                        & 4805.39                        & $34\times$                      \\

$2^{26}$                        & 2169.02                       & 76900*                          & $35\times$                      \\
$2^{30}$                        & 34022.77                      & 1230200*                        & $36\times$                               \\ \bottomrule
\end{tabular}
\caption{Comparison of communication for SHA3-256 baseline and our scheme for executing the $\assert$ functionality. Here $p_c = 1/2$. * means estimated via extrapolation}
\label{tab:comm}
\end{table}

\section{Number of Indices $|\mathcal{I}|$}\label{app:numind}
\begin{figure}[!htbp]
	\centering
		\includegraphics[scale=0.5]{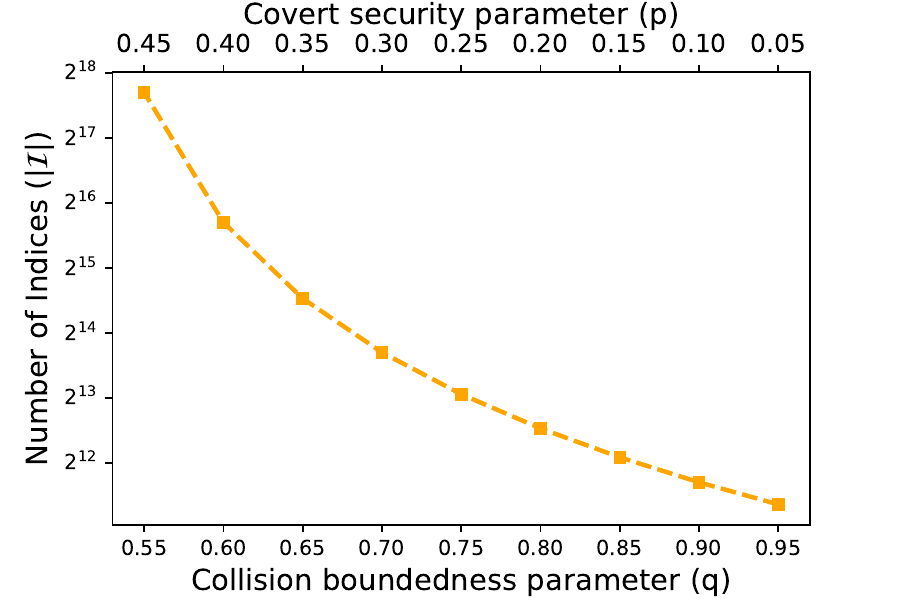}
		\caption{Number of indices (hashes) $|\mathcal{I}|$ needed to be computed by the committer as a function of $q$ (and the covert security parameter $p$ i.e. $1-q$). Here block size $b = 1024$ and statistical security parameter $\sigma = 40$.}
		\label{fig:indices}
\end{figure}

We use the formulation, upon ceiling to the next nearest integer, defined in Theorem~\ref{thm:construction3} to compute $|\mathcal{I}|$. In Figure~\ref{fig:indices}, we plot this formulation for $\sigma = 40$, $b = 1024$ and different values of q (and the covert security parameter $p$ i.e $1-q$) to show how the number times $|\mathcal{I}|$ that the committer needs to compute $H$ varies with the security parameters.

\section{Certified Predictions}\label{app:fairness}

\begin{table}[t!]
\centering
\resizebox{\columnwidth}{!}{
\begin{tabular}{lccc|ccc}
\hline
        & \multicolumn{3}{c|}{Accuracy}  & \multicolumn{3}{c}{Average Odds Difference} \\ \hline
Dataset & Unfair   & Fair     & Changed  & Unfair        & Fair         & Changed      \\ \hline
Credit  & $\mathrm{69.3\%}$ & $62.7\%$ & $64.3\%$ & $-0.341$      & $0.359$      & $0.122$      \\
COMPAS  & $67.6\%$ & $55.9\%$ & $65.0\%$ & $-0.181$      & $0.369$      & $0.073$      \\
Adult   & $80.4\%$ & $75.8\%$ & $78.8\%$ & $-0.270$      & $0.261$      & $0.111$      \\ \hline
\end{tabular}
}
\caption{Accuracies of a fair prediction method \cite{zhang2018mitigating} (Fair), the same model changed by a single weight to maximize accuracy (Changed), compared to the model trained without any constraints (Unfair).}
\label{table.motivation}
\vspace{-4ex}
\end{table}

In this section we describe how PVC committed MPC enables a key application, \emph{certified predictions}: obtaining secure predictions by a private model that is certified to have certain properties (more on such properties below). We show by means of a real-world example
how heuristic approaches that are sublinear in the input size fail. We do this by training a fair machine learning model and showing how,
by modifying a single parameter of the model, it can be made unfair and more accurate. We describe previous work on the problem of obtaining predictions
by a certified private model, and discuss 
an efficient solution enabled by our results.

Recently, ML models have started to be deployed into high-impact, real-world decision-making settings such as medicine \cite{davenport2019potential}, self-driving cars \cite{bojarski2017explaining}, and college admissions \cite{putra2018implementation}. However, this has led to problems: many of these settings have key constraints that ML models were not originally designed to handle. Current models lack \emph{interpretability} \cite{lipton2018mythos}, \emph{safety} \cite{amodei2016concrete}, and \emph{fairness} \cite{obermeyer2019dissecting}. To address this, there has been a wealth of recent work aimed at formalizing these constraints and creating ML models that satisfy them \cite{ribeiro2016should, mohseni2019practical, madras2018learning, celis2019classification, zhang2018mitigating}. 

However, models that satisfy these constraints often have reduced accuracy as the constraints restrict the model's predictions in accuracy-agnostic ways. As model accuracy is often directly tied to beneficial outcomes (e.g., monetary investment, company profit, likelihood of publication), real-world constraints can incentivise service providers to cheat. 
To prevent this a natural question arises: \emph{What is the minimal computation required to ensure cheating does not occur?} One may be tempted to try to construct a procedure that is sublinear in the size of the model. Recent work has proposed to generate a small series of tests to identify small changes to a model \cite{he2018verideep}. However, we show with a simple example that any protocol must ensure that nothing about the model is changed, requiring a linear time procedure.

\textbf{Any change may sacrifice fairness.}
We investigate a popular real-world constraint placed on ML models: \emph{fairness constraints}. In general, the most popular formulation of fairness constraints minimizes the difference between (functions of) predictions made on different demographic groups. Because these techniques constrain predictions across groups, their accuracy is less than unconstrained models. We investigate a popular fair prediction model \cite{zhang2018mitigating} applied to three fair prediction problems: judging credit risk (Credit\footnote{https://tinyurl.com/cm-credit}); predicting parole violators (COMPAS\footnote{https://tinyurl.com/cm-compas}); inferring income (Adult\footnote{https://tinyurl.com/cm-census}). We consider the following \emph{average odds difference} fairness criterion
\begin{align}
\Big( \mathbb{E}[\hat{Y} \mid A\!=\!0, Y\!=\!y] - \mathbb{E}[\hat{Y} \mid A\!=\!1, Y\!=\!y] \Big) \geq \tau, \;\;\; \forall y \in \{0,1\}, \nonumber
\end{align}
where $Y$ is the true outcome (e.g., $Y\!=\!1$ signifies good credit in Credit, while $Y\!=\!0$ signifies bad credit) and $\hat{Y}$ is the prediction. Here $A$ indicates demographic group (e.g., race, gender, sexual orientation, among others). Specifically $A\!=0\!$ indicates a \emph{disadvantaged group} and $A\!=1\!$ indicates a \emph{privileged group}. Thus the above constraint says that the average outcome for the disadvantaged group has to be at least $\tau$-larger than the average outcome for the advantaged group. This is to combat predictors $\hat{Y}$ that benefit the privileged group (such predictors will arise from unconstrained training). These expectations are computed over a training dataset. Table~\ref{table.motivation} shows the accuracy and average odds difference of the model in \citet{zhang2018mitigating} using the fairness constraint (Fair), compared to the model without the fairness constraint (Unfair).

Now we imagine that a cheating service provider wants to take the fair model and only change a single element of the model to maximize accuracy. We imagine they test every single element, optimizing for accuracy alone, while fixing the remaining parameters. They then take the model which has the maximum improvement in accuracy across all single-parameter-changed models (Changed). We report the accuracy and fairness of this model in Table~\ref{table.motivation}. 

These results show that changing just a single element can significantly improve the accuracy over the fair model (by as much as $9.1\%$ on COMPAS). Further, the changed model has significantly lower average odds difference than the fair model, unfairly benefiting the privileged group at the expense of the disadvantaged group. Thus, to ensure a service provider cannot surreptitiously improve accuracy at the expense of real-world constraints, a protocol must ensure that the entire model remains the same.

\textbf{Related work.}
To prevent this a number of works have proposed techniques to verify ML models \cite{ghodsi2017safetynets,he2018verideep,kilbertus2018blind,segal2020fairness}. SafetyNets \cite{ghodsi2017safetynets} propose an interactive proof protocol for verifying deep neural network predictions. This protocol only has a verification guarantee and leaves a security guarantee to future work. Further it is limited to models expressible as arithmetic circuits. VerIDeep \cite{he2018verideep} describe a method to generate inputs for which small changes to the ML model would yield very different outputs. However, this model does not guarantee that the entire model remains the same and thus would be vulnerable to attacks similar to that described above.

Recent work with security guarantees \cite{kilbertus2018blind,segal2020fairness} propose to use hash functions (SHA-256, SHA-3 in sponge mode) to verify a model has not been altered. Specifically these works generate and verify a hash within MPC. In MPC the protocol cost is dominated by AND gate computations and the most efficient method requires asymptotically 35 AND gates per input bit \cite{segal2020fairness}. While there exist an MPC-optimized hash called LowMCHash-256 \cite{albrecht2015ciphers} it is new and susceptible \cite{dinur2015optimized}. Our constructions above enable secure predictions with verified inputs that asymptotically require 0.5 AND gates per input bit and derives security from the well-known random oracle assumption.

\textbf{Our approach.}
To enable certified predictions we propose the following procedure. First the service provider (committer, P1) makes a commitment $c = \pvccommit(x,\sk,r)$ to a model $x$. P1 then engages in an MPC protocol with a regulatory agency (P2') where P2' verifies the model $x$ satisfies the required guarantee (e.g., fairness), and that $c$ is a commitment to that model. If these checks pass then P2' signs the commitment $c$ with their private key and sends it to P1. When a user (verifier, P2) wishes to obtain a certified prediction from P1, they engage in a PVC commitment. Here P1 sends $c$ to P2. If (a) P2 can verify that $c$ is signed by the regulatory agency P2' (e.g., this could be done if regulator's public key is publicly available) and (b) the PVC commitment is verified (via $\assert$ and $\checkfun$ as described in Figure~\ref{fig:basic_diagram}), then the output is a certified prediction.

\section{Proofs and Definitions of Section~5 \\ (Integrity Checking)}\label{app:ic}
\subsection{Proof that we have constructed a PVC \\ commitment scheme}\label{app:proofcommitmentscheme}

\secpvccom*
\begin{proof}
To show that the functions $\pvccommit$, $\assert$ and $\checkfun$ for a PVC commitment scheme with parameter $p$, we must check that $\checkfun$ is deterministic and that the four properties hold.

As the function $\checkfun$ is given by a decision tree depending on checking whether (deterministic) parts of the input are equal and whether signatures are valid (which is deterministic by the assumption on the verification function) it is deterministic.

\paragraph{Correctness}
Given $i,r,x$ and a valid key pair $\sk,\pk$, let $c=\pvccommit(x,\sk,r)$ and $a=\assert(x,\sk,r;i,\pk)$. Consider the definition of $\checkfun(c,a,\pk)$. As the key pair is valid both of the signatures will check out thus the result is not $\inconclusive$. Furthermore, both $c[a[0]]$ and $a[1]$ are equal to $H(i,r,x)$, thus the check will return $\valid$.

\paragraph{General Binding}
Suppose, these functions do not satisfy general binding. Then there exists a polynomial time adversary, $\adv$, contradicting Definition~\ref{def:generalbinding}. Let $x,\sk,r,x',\sk',r',\pk$ and $c$ be the output of this adversary. Further, let $i,a,a',\out$ and $\out'$ be as in the definition. As the signature scheme is discrimination resistant the distribution of $i$ conditioned on $G$ (and thus on $\out$ or $\out'$ being $\inconclusive$) is still uniform. It follows that in order for Inequality~\ref{eqn:generalbinding} to hold we must have
\begin{equation*}
\PP(H(i,r,x)\neq c[i])+\PP(H(i,r',x')\neq c[i])<p.
\end{equation*}
Thus for greater than a $1-p$ fraction of the choices of $i$ we must have $H(i,r,x)=c[i]=H(i,r',x')$. This would mean that $r,x,r',x'$ is a $q'$-collision for some $q'>q$. The above process then gives a polynomial time algorithm contradicting the $q$-collision boundedness of $H$. So the general binding property must hold.

\paragraph{Hiding}
Suppose $\adv$ is a polynomial time adversary contradicting Definition~\ref{def:schemehiding}. Consider the polynomial time algorithm that takes as input $(H(i,r,x))_{i\in\mathcal{I}}$, computes $\adv(\mathcal{O}_\sk(x))$ (using a hard-coded $\sk$) and outputs the result. This adversary contradicts the hiding property of $H$ (Definition~\ref{def:hiding}).

\paragraph{Defamation Freeness}
Let $\adv$ be a polynomial time adversary. In order to have $\checkfun(c,a,\pk)=\cheated$ both $c$ and $a$ must be correctly signed. As the signature scheme is chosen-plaintext secure $\adv$ can only achieve this with non-negligible probability by using the contents of $\mathcal{O}_\sk(x)$ as $c$ and $a$ (they ca not even be switched as they have different formats). But with that choice of $c$ and $a$, $c[a[0]]=a[1]$, and thus the check would return $\valid$. Therefore the functions are defamation free.
\end{proof}

\subsection{Execution in the ideal world}\label{sec:ideal}

Next, we present in detail the ideal world execution
of a function $g(x,y)$. The ideal world is parameterized by the party corrupted by
adversary $\adv$,
which we denote by $\corrupted\in \{\A, \B, \bot\}$ ($\bot$ is just a value different from $\A$ and $\B$ to represent that all parties are honest) and, as mentioned above, two
probabilities $p_{\texttt{exec}}, p_{\texttt{commit}}$.
Let us remark that $\adv$ has an auxiliary input, and that all parties are initialized with the same value
on their security parameter tape (including the trusted party), but we leave both of these aspects implicit for clarity.

An unusual aspect of this ideal world is the presence of an ``observation of the environment''
which happens after the commitment has been made but before the computation.
The idea being that a party has committed to an input
if they are unable to make it depend on something they learnt between commitment and computation.
We assume that this observation is drawn from some distribution $\mathcal{E}$
and that the distribution can be sampled from by a polynomial time algorithm.
This latter assumption is to stop the environment from encoding, say, collisions of zero for a secure hash function.

\paragraph{1. $\A$ receives input.}
Party $\A$ receives its prescribed input $x$.
\paragraph{2. Commitment of $\A$'s input.}
At this stage $\A$ sends to the 
trusted party the input that it intends 
to use in a subsequent computation, 
which we denote by $w$.
If $\corrupted \neq \A$, then
$w = x$, and otherwise
$\adv$ sets $w$ to be an arbitrary
valid input value in a way that might depend on $x$.

\paragraph{3. The environment is revealed}
The value $e$ is sampled from the distribution $\mathcal{E}$
Party $\A$ is given the value of $e$

\paragraph{4. $\B$ receives input and parties send inputs.}
Party $\B$ receives its prescribed input $y$.
Next, $\A$ and $\B$
send to the trusted party their inputs to be used in the computation,
denoted $a,b$, respectively.
\begin{itemize}
\item If $\corrupted = \A$, then $\adv$ sets
  
$a\in \{w, \texttt{abort}, \texttt{corrupted},
\texttt{cheat\_exec},
\texttt{cheat\_commit}\}$ for $\A$, and otherwise
$\A$ sets $a = w$.

\item If $\corrupted = \B$ then $\adv$ sets $b$ for $\B$,
otherwise
$\B$ sets $b = y$.
\end{itemize}

\paragraph{5. Early abort \& blatant cheating.}
$\adv$ is given the opportunity 
to have $\corrupted$ abort or announce that it is corrupted.
This results in updating either $a$ (if $\corrupted = \A$) or $b$ (if $\corrupted = \B$) to $\texttt{abort}$ or $\texttt{corrupted}$.
If that is the case, $a$ (resp. $b$) is forwarded to $\B$ (resp. $\A$) and the trusted party halts.
\paragraph{6. Attempted cheat option.}
If $a = \texttt{cheat\_exec}$, then
the trusted party tosses 
a coin $X = \texttt{Ber}(p_{\texttt{exec}})$, where $p_{\texttt{exec}}$ is the probability of $\A$
getting caught cheating at this stage, and
\begin{itemize}
    \item If $X = 1$ then the trusted party sends $\texttt{corrupted}$ to both $\A$ and $\B$.
    \item If $X = 0$ then the 
    trusted party sends $\texttt{undetected}$
    to $\A$, along with $y$ ($\B$'s input). Following this, $\adv$ gets to choose $\B$'s output of the protocol, and sends it to the trusted party.
\end{itemize}
The ideal execution ends at this point if $a = \texttt{cheat\_exec}$.
\paragraph{7. Attempted break commitment option.} If
$a = \texttt{cheat\_commit}$
then
the trusted party requests from $\A$ (a) a probability $q$ and (b) a new value $w'$ for $w$.
The trusted party then sets $p = q$, if $w = w'$, and $p = \max(q, p_{\texttt{commit}})$
otherwise, where $p_{\texttt{commit}}$ is the probability of $\A$
getting caught cheating at this stage. (Note that this simply allows the adversary to
choose an arbitrary probability of getting caught when cheating to rewrite $w$
with the same value again).
Then, the trusted party
(i) tosses 
a coin $Y = \texttt{Ber}(p)$, (ii) rewrites $w$ to take value $w'$,
(iii) runs $g_1(w, b)$ with the updated $w$, and
(iv) gives $\adv$ the opportunity to abort $\A$. Next,
\begin{itemize}
    \item if $Y = 1$ then the trusted party sends $\texttt{corrupted}$ to both $\A$ and $\B$ and halts, and
    \item if $Y = 0$ then the 
    trusted party sends $\texttt{undetected}$
    to $\A$.
\end{itemize}

Let us remark that giving $\adv$ the opportunity to abort upon observing the output in in step 7.
is allowed just to simplify the presentation of out protocol, and that
an extension where $\adv$ does not receive an output when caught cheating is
easy to achieve by just adding a round of interaction to our protocol.
In that extra round $\B$ enables $\A$ to ungarble their output after verifying
the commitment resulting from the secure computation.

\paragraph{8. Trusted party gives out outputs.}
The trusted party evaluates $g(w, b)$, and gives
$\adv$ the chance to abort the execution. Otherwise it
gives their designated output to $\B$, at which point $\adv$
is allowed to either abort the execution, or let the honest party receive their output.

\paragraph{Outputs.} The honest party  outputs what they received in the final step, and 
$\adv$ outputs an arbitrary (probabilistic) polynomial-time computable function of 
$\corrupted$'s input, any auxiliary input, and its view during the execution.

\subsection{Definitions}\label{sec:icdefs}

The following simulation security definition deviates from most such definitions in that we allow the adversary in the ideal model to be logically omniscient, whereas it is standard to restrict the simulator to polynomial time computations. The polynomial time assumption is important in the context of zero-knowledge proofs and for certain systems of composability. However standard bit commitment is impossible in the universal composability model~\cite{CF01}, thus we must settle for weaker composition guarantees here. We hope it is clear that the ideal setting here is information theoretically secure. Thus even a logically omniscient adversary can not possibly learn things that it should not in the ideal model.
The simulator could be made computable at the expense of slightly complicating the proof, but as this is unnecessary and also not standard we prefer to keep the proof simple.

Denote by $\texttt{IDEAL}^{p_{\texttt{exec}}, p_{\texttt{commit}}, \mathcal{E}}_{g,\mathcal{S}(z),i}(x,y,\lambda)$ the environment variable and the outputs of the honest parties and adversary in an execution in the ideal world above, and let $\texttt{REAL}^{p_{\texttt{exec}}, p_{\texttt{commit}}, \mathcal{E}}_{g,\adv(z),i}(x,y,\lambda)$ denote the environment variable and the outputs of the honest parties and the adversary in a real execution of a protocol $\pi$.

\begin{definition}
Let $g$ and $p$ be as above. A protocol $\pi$ \emph{securely computes $g$ with committed first input in the presence of a malicious P1 or a covert P2 with $p$-deterrent} if for every non-uniform probabilistic polynomial time adversary $\adv$ for the real model, there exists a definable adversary $\mathcal{S}$ for the ideal model such that for each $i \in \{1,2\}$:
\begin{equation*}
  \begin{split}
  &\left\{\texttt{IDEAL}^{p_{\texttt{exec}}, p_{\texttt{commit}}, \mathcal{E}}_{g,\mathcal{S}(z),i}(x,y,\lambda) \right\}_{x,y,z\in \{0,1\}^*, \lambda\in\mathbb{N}}  \equiv^c \\
& ~~~~~~\left\{\texttt{REAL}^{p_{\texttt{exec}}, p_{\texttt{commit}}, \mathcal{E}}_{g,\adv(z),i}(x,y,\lambda) \right\}_{x,y\in \{0,1\}^*, \lambda\in\mathbb{N}}
    \end{split}
\end{equation*}
\end{definition}

In PVC security it is important that a fail-stop adversary is not labelled a cheat (at least in most contexts including ours) for that we say that:

\begin{definition}
A protocol $\pi$ is \emph{non-halting detection accurate} if for every fail-stop adversary $\adv$ controlling party P1, the probability of an honest P2 outputting $\texttt{corrupted}$ is negligible.
\end{definition}

In order to have PVC security in place of the covert security we require that their be some algorithm $\texttt{Blame}$. When applied to the view of an honest party that has outputted $\texttt{corrupted}$, it must return a proof of that corruption. The proof is verified by another algorithm $\texttt{Judgement}$, which will output $\cheated$ if and only if it is a genuine proof. These ideas are formalized as follows.

Given an algorithm $\texttt{Commit}$ and a protocol $\mathcal{P}$ let the commitment protocol formed by them consist of P1 running $\texttt{Commit}$ and sending the result to P2, P1 then receiving $e\gets \mathcal{E}$ and then both P1 and P2 engaging in $\mathcal{P}$ and taking the output from that protocol as the output.

\begin{definition}\label{def:main}
A quadruple $(\texttt{Commit},\mathcal{P},\texttt{Blame},\texttt{Judgement})$ \emph{securely computes $g$ with committed first input in the presence of a malicious P1 or a covert P2 with $p$-deterrent and public verifiability} if the following hold:
\begin{enumerate}
\item (Simulatability with $p$-deterrent:) The commitment protocol formed from $\texttt{Commit}$ and $\mathcal{P}$ securely computes $g$ with committed first input in the presence of a malicious P1 or a covert P2 with $p$-deterrent and is non-halting detection accurate.
\item (Accountability:) For every PPT adversary $\adv$ controlling P1 and interacting with an honest P2,
\begin{equation*}
\PP(\B\mbox{ outputs }\texttt{corrupted} \wedge \texttt{Judgement}(\texttt{Blame}(\texttt{View}(\B)))) \neq \cheated)
\end{equation*}
is negligible.
\item (Defamation-Free:) For every PPT adversary $\adv$ controlling P2 and interacting with an honest P1,
\begin{equation*}
\PP(\adv \mbox{ outputs } \wedge \texttt{Judgement}(Cert) = \cheated)
\end{equation*}
is negligible.
\end{enumerate}
\end{definition}

\subsection{PVC committed MPC Proof}\label{sec:mainproof}

\mainthm*
\begin{proof}

\myparagraph{Simulating P1}
First we consider simulatability in the case where P1 is corrupted.
Given a non-uniform probabilistic polynomial time adversary $\adv$, we construct $\mathcal{S}$ as follows.

First note that $\mathcal{S}$ can uniformly randomly choose a randomness tape which will be used for all of its black box runs of $\adv$, this reduces the task to the case where $\adv$ is a deterministic adversary.

By running $\adv$ the simulator is given the commitment $c$ that $\adv$ chooses to use (which may or may not be generated by applying $\texttt{Commit}$ to some $w$). Now $\mathcal{S}$ can look at how $\adv$ would respond to every possible environment variable $e$ (this is fine because it is a mathematical function which need not be computable).
If, in response to $e$, $\adv$ does any of early abort, blatant cheat or cheat during the execution the outcome is independent of what commitment was made so it would not matter what $\mathcal{S}$ commits to in the ideal world.
The other possibility is that $\adv$ does none of those things and submits some $x'$ as their input alongside supposed randomness $r'$.

The underlying protocol $\Pi$ allows input extraction in polynomial time (as is used in the proofs of security for that protocol in Hong et. al.~\cite{pvc}) thus in all of these other cases $\mathcal{S}$ can extract which $x'$ will be used in response to each $e$.
For each one $\mathcal{S}$ can then compute $a=\assert(x',\sk,r';i,\pk)$ for every $i\in \mathcal{I}$ and check for what fraction of the $i$ we have $\checkfun(c,a,\pk)=\cheated$.

For those $e$ that result in being caught with probability at least $p/2$ it would not matter what $\mathcal{S}$ committed to as it will be able to attempt to cheat the commitment to change the input to $x'$ and get caught with the correct probability. After which it will receive $g_1(x',y)$, add it to the simulated view, and proceed according to what $\adv$ would do next. Aborting if and only if $\adv$ chooses to abort. P2 will then receive $\texttt{corrupted}$ with the correct probability.

Those $e$ that result in less than $p/2$ probability of being caught, it will matter that $\mathcal{S}$ commited to the value of $x'$ that $\adv$ wants to use. Thus for the commitment phase $\mathcal{S}$ will commit with the trusted party to the value $\tilde{x}$ that is most likely to be used as $x'$ (with respect to the randomness of $e$).
If the adversary uses $x'=\tilde{x}$ then $\mathcal{S}$ will now be able to tell the trusted party it wants to use that value and it wants to get caught with the correct probability.

The remaining possibility, that $e$ results in $\adv$ using an input $x'\neq \tilde{x}$ and $r'$ which has a probability less than $p/2$ of resulting in P1 being caught, would be a serious problem for $\mathcal{S}$.
We claim however that this can happen with only negligible probability.

Suppose to the contrary that some non-negligible fraction of the weight of $\mathcal{E}$ resulted in these bad $x',r'$. Then as each must individually have weight at most that assigned to $\tilde{x}$ we can construct a polynomial time algorithm as follows.

Sample $e\gets \mathcal{E}$ and extract the input of $\adv$ for this $e$, compute $a=\assert(x',\sk,r';i,\pk)$ for all $i\in \mathcal{I}$, check to see if less than a $p/2$ fraction of the $i$s would result in $\checkfun(c,a,\pk)=\cheated$. Repeat this process until two distinct such values of $(x',r')$ have been found with this property. As the fraction of $e$ that result in finding such an $x'$ other than the most common one is non-negligible, this algorithm runs in expected polynomial time.

However, this can (by putting a polynomial time upper bound on the run time and failing if it reaches it) be converted into a PPT algorithm which contradicts the general binding property of the PVC commitment scheme. Thus proving the claim.

This addresses simulating the correct distribution between $e$, the output of P2 and messages explicitly sent in our protocol to P1. The messages sent to P1 in the secure computation black box are dealt with by the simulator for $\Pi$ as given in Hong et. al.~\cite{pvc}.

Extending this to the optimized integrity check case is straight forward, everything is the same except $\mathcal{S}$ must produce a fake hash-commitment for $\adv$ to sign. This can be done by hashing randomness due to the hiding property of the commitment scheme this would not break indistinguishability. Further whilst the signed version should be given to P2 in the real world it does not form part of P2's output so we need not worry about coordinating with that.

If the commitment had been produced honestly, then $\adv$ must have some $x,r$ that it committed to that it knows of. This together with any value of $x',r'$ that collides with the resulting commitment $c$ less than some fraction $p$ of the time will break general binding. Thus the same simulator as above with this extra observation gives the stronger security.

\myparagraph{Simulating P2}
Simulating the other side is much easier. P2 receives a commitment $c$ to some $x$, however due to the hiding property of the commitment $\mathcal{S}$ can get away with providing $\adv$ with a commitment to some arbitrary value, say $0$. The simulator can now have $\adv$ interact with a copy of P1 (multiple times) in order to extract their input $y$ and index $i$. It can then give this input to the trusted party to find the correct value of $g_2(x,y)$. It can then add the result of $g_2(x,y)$ to the simulation of the adversaries view using the simulator for $\Pi$.

With the optimization the only change is that rather than giving the value of $\assert$ for the given $i$, a pair $(i,c[i])$ is signed and added to the simulated view.

\myparagraph{Correctness}
If both parties are honest then $\Pi$ will correctly output $(g_1,(x,y),g_2(x,y))$ and the result of $\assert$. As $\assert$ is computed correctly on the honest inputs the check with the commitment from the first step will return $\valid$. And thus P2 will accept the output and the parties will have successfully computed $g$. 

\myparagraph{Accountability}
For accountability, we need to show that a cheating $\A$
gets caught \emph{publicly}, in the sense that if $\B$ claims that $\A$
cheated then there's a proof accompanying that claim except with negligible probability.
Note that if $\A$ cheats inside
the secure computation with $\Pi$, this follows from the PVC properties of $\Pi$, and the definition
of $\texttt{Blame}$ and $\texttt{Judgemenet}$  in terms of
$\texttt{Blame}_{\texttt{exec}}$ and $\texttt{Judgement}_{\texttt{exec}}$.
If $\A$ cheats in that the input to $\Pi$ differs from the committed value
then $\texttt{Judgement}_{\texttt{commit}}(\texttt{Blame}_{\texttt{commit}}(\cdot))$
will output $\cheated$ after verifying that
$\texttt{Blame}_{\texttt{commit}}(\cdot)$ constitutes a valid signature of
the fact that the commitment $c$ and evaluation of $H$ at $i$ do not match.
This happen with all but negligible probability due to the properties of the
cryptographic signature, and the correctness of the PVC commitment scheme, i.e.
different values for indexed hashes necessarily come from different inputs.

\myparagraph{Defamation Freeness}
Defamation freeness states that proofs of cheating can not be forged. This holds
for proofs outputted by $\Pi$ by the fact that it satisfies PVC security,
and it holds for proofs generated by
$\texttt{Blame}_{\texttt{exec}}$
due to the properties of the
cryptographic signature scheme.

\end{proof}

\section{Proofs of Section~7 (Lower Bounds)}\label{app:lb}
\lowerpropone*
\begin{proof}
  We give a polynomial time algorithm which, given a circuit, $C$, that implements a function from $\{0,1\}^a \times \{0,1\}^n$ to $\{0,1\}^m$ and contains fewer than $\lceil (n-m)/2 \rceil$ nonlinear gates, returns a non-zero input $x \in \{0,1\}^n$ such that $C(s,x)=C(s,0)$ for all $s\in \{0,1\}^a$. As this algorithm finds a $1$-collision with certainty if the circuit is small enough, for $H_k$ to be secure that must happen with negligible probability in $\lambda$. And the result is immediate.

  Consider the wires of $C$ that are either outputs of the circuit or inputs to nonlinear gates. The hypotheses imply that there are $<n$ of such wires. Wire $j$ must contain the XOR of a linear (i.e. parity) function $f_j$ of the input with an affine function of the key and nonlinear gate outputs (either of which could be trivial).

  The conditions $f_j(x)=f_j(0)$ form a collection of $<n$ linear constraints in $n$ variables. Since $x=0$ is obviously a solution of this under-determined system, then it must have also a nontrivial solution, which can be found efficiently.
\end{proof}

\lowerproptwo*
\begin{proof}
  The idea of the proof is similar to that of Proposition~\ref{prop:lowercovert}.
  We give a polynomial time algorithm which given a circuit $C$ from $\{0,1\}^n$ to $\{0,1\}^m$ with fewer than $n-m$ nonlinear gates returns a collision in that circuit. Thus to have collision resistance the circuit can be that small only with negligible probability.

  Let $d$ be the number of nonlinear gates in $C$, we will assume WLOG that they are AND gates. We will derive $d+m$ affine conditions on $x$ which determine $C(x) = C(0)$. As affine systems are efficiently solved, collision resistance requires that there be at most one solution ($x = 0$) to this system. For this to happen $d+m\geq n$ must hold, which proves the statement.

  Thus it remains to derive the aforementioned $d+m$ affine conditions. We need to fix the output of each AND gate using only one affine condition on $x$. This can be achieved as follows. Consider each AND gate in (a totalisation of the partial) order from input to output, i.e. a topological ordering of the circuit. For each AND in that sequence, if the first input can be set to $0$ with an affine condition then add that condition to the set and move on to the next gate. Otherwise, the first input is already determined so we need only add an affine condition that fixes the second input. Either way that is only one condition per gate. In summary, by adding one condition for each of the output wires we can determine their values, so we are done.
\end{proof}

\lowerpropthree*

\begin{proof}
  For a indexed hash function circuit $C$ with $n$-bit main input and $d$ nonlinear gates we explain how to do each of the following in time polynomial in the size of $C$:
  \begin{itemize}
  \item Transform $C$ into another circuit $\tilde{C}$ with $d$ non-linear gates.
  \item Simulate the output of $\tilde{C}(i,r,x)$, given the output of $C(i,r,x)$ and $i$.
  \item Derive a collision in $\tilde{C}$ from a collision in $C$
  \end{itemize}
  Finally we will show that $\tilde{C}_k$ has output length at most $3d_k$ with all but negligible probability.
  It follows that $\{\tilde{C}_k\}_{k\in K}$ is non-trivially collision bounded and the result follows from Propositions~\ref{prop:lowercovert} and~\ref{prop:lowerordinary}.
  
  Throughout this proof $L$ with a subscript will denote a linear function.

  Let $C$ be a circuit with input $(i,r,x)$ where $x$ is the $n$-bit input, $r\in \mathcal{R}$ is the randomness and $i$ is an index. All the following computations can be done in polynomial time we will avoid repeating this fact for each one.

  Note that $C(i,r,x)$ can be rewritten as $L_1(i,r,x,g(L_2(i,r,x)))$ for some nonlinear function $g$ where $L_2$ has a $2d$ bit output and $g$ has a $d$ bit output and is implemented with $d$ non-linear gates.

  Considering $L_2$ as a linear function of $\mathcal{R}$ we can find its kernel $T$ which has codimension at most $2d$. Compute representations of $\pi_{T^\perp}(r)$ and $\pi_{T}(r)$, represented in a basis of $T^\perp$ and a basis of $T$, call them $r_1$ and $r_2$ respectively. Thus $r_1$ has length at most $2d$, and $L_2(i,r,x)$ is equal to some $L_3(i,r_1,x)$. We can thus write $C(i,r,x)$ as
  \begin{equation}\label{eq:brokendown}
    L_4(i)+L_5(r_2)+L_6(r_1,x,g(L_3(i,r_1,x)))
  \end{equation}
  We now define $\tilde{C}(i,r_1,x)$ to be a representation of
  \begin{equation}
    \pi_{(\text{Im} L_5)^\perp}(C(i,r,x)-L_4(i))=L_7(x,r_1,g(L_3(i,r_1,x)))
  \end{equation}
  in a basis of $\text{Im} L_7$.

  As $\tilde{C}(i,r,x)$ is a linear function of $C(i,r,x)$ and $i$ we can write it with $d$ nonlinear gates.

  As $\tilde{C}(i,r_1,x)$ is a known linear function of $i$ and $C(i,r,x)$ so simulating it is trivial.

  We abuse notation and use $r$ for the function that recovers $r$ from a derived $r_1$ and $r_2$. Suppose that $\tilde{C}(i,r_1,x)=\tilde{C}(i,r'_1,x')$ and $x\neq x'$. Then we can compute $C(i,r(r_1,0),x)-C(i,r(r'_1,0),x')$ which by the definition of $\tilde{C}$ will be in $\text{Im} L_5$, we can then choose $r_2$ such that
  \begin{equation}
    C(i,r(r_1,0),x)-C(i,r(r'_1,0),x')=L_5(r_2)
  \end{equation}
  which combined with Equation~\ref{eq:brokendown} yields
  \begin{equation}
    C(i,r(r_1,0),x)=C(i,r(r'_1,r_2),x')
  \end{equation}
    
  Finally, recall that $\tilde{C}(i,r_1,x)$ is a representation of $L_7(x,r_1,g(L_3(i,r_1,x)))$ and that $\tilde{C}$ has full rank. Thus we can write each $\tilde{C}_k$ as
  \begin{equation}
    L^k_8(x)+L^k_9(r_1,g(L_3(i,r_1,x)))
  \end{equation}
  As $(r_1,g(L_3(i,r_1,x)))$ is at most $3d$ bits long the rank of $L^k_9$ must be at most $3d$. If $\text{Im} L^k_8$ is contained in $\text{Im} L^k_9$ then the length of the output of $\tilde{C}_k$ is at most $3d$. Otherwise, $\pi_{(\text{Im} L^k_9)^\perp}\tilde{C}_k(i,r,x)$ is a non-trivial linear function of $x$, however this latter possibility must occur with negligible probability otherwise $\{\tilde{C}_k\}_{k\in K}$, and thus $\{C_k\}_{k\in K}$, is not hiding.
\end{proof}

\section{Proofs of Section~8 (from Covert to \\ Malicious)}\label{app:malicious}
\maliciousone*
\begin{proof}
Suppose $x$ and $x'$ are two distinct inputs. Then for some $j$, $x^j\neq x'^j$. Then by the error detecting property $E(x^j)$ and $E(x'^j)$ must differ in at least $\kappa$ places, $j_1,...,j_\kappa$. Thus $\tilde{x}_{j_m}\neq \tilde{x}'_{j_m}$ for all $m\in \{1,...,\kappa\}$.

Suppose that $r,x,r',x'$ is a $q'^\kappa$-collision of $H_k^E$ for some $q'>q$. Then at least one of $r, \tilde{x}_{j_m}, r', \tilde{x}'_{j_m}$ is a $q'$-collision of $H_k$. Thus, by testing each $m$ in turn, a $q'$-collision of $H_k$ can be found in polynomial time from a $q'^\kappa$-collision of $H_k^E$. The first part of the result follows. The number of required gates follows from counting through the algorithm for $H^E$.
\end{proof}

\section{Proofs of Section~9 (Arithmetic \\ Circuits)}\label{app:arithmetic}
\arithmeticthm*
\begin{proof}
By Theorem~\ref{thm:collision} it suffices to show that for two distinct inputs $x,x'$ at most $b$ values of $i$ will result in $d_4(i,x)=d_4(i,x')$.
Consider two distinct inputs $x, x'$, and assume WLOG that they differ in the first $b$ field elements. Let $y$ and $y'$ be the first $b$ field elements from $x$ and $x'$ respectively.

It suffices to show that only $b$ values of $i$ will result in $d_4(i,y)=d_4(i,y')$. Start by expanding out the difference.
\begin{equation*}
\begin{split}
  d_4(i,y)-d_4(i,y')=\sum_{j=1}^{b/2} \big( &i^{2j-1}(y_{2j}-y'_{2j})~+\\
  & i^{2j}(y_{2j-1} -y'_{2j-1})~+\\[1.4ex]
  & y_{2j-1}y_{2j} - y'_{2j-1}y'_{2j}\big)
\end{split}
\end{equation*}
Let $s$ be the function on the natural numbers that swaps $2j$ and $2j-1$ for all $j$. The difference is given by the inner product of $(1,i,i^2,...,i^b)$ with a vector $v(y,y')$ in $\mathbb{F}^{b+1}$ with zeroth entry being
\begin{equation}
  \sum_{j=1}^{b/2} y_{2j-1}y_{2j} - y'_{2j-1}y'_{2j}
\end{equation}
and $j$th entry for $j>0$ being $y_{s(j)}-y'_{s(j)}$.

We have $d_4(i,y)=d_4(i,y')$ only if $(1,i,...,i^b)$ is perpendicular to $v(y,y')$. But as any $b+1$ vectors of the form $(1,i,...,i^b)$ form a Vandermonde matrix and thus are linearly independent at most $b$ of them could be perpendicular to any fixed $v(y,y')$.
\end{proof}

\end{document}